\documentclass[onecolumn]{IEEEtran}
\usepackage[mathscr]{eucal}
\usepackage[cmex10]{amsmath}
\usepackage{epsfig,epsf,psfrag}
\usepackage{amssymb,amsmath,amsthm,amsfonts,latexsym}
\usepackage{amsmath,graphicx,bm,xcolor,url,overpic}
\usepackage{fixltx2e}%ordering of single and double column floats
\usepackage{array}%array and tabular environments
\usepackage{verbatim}
\usepackage{bm}
\usepackage{algorithmic}
\usepackage{algorithm}
\usepackage{verbatim}
\usepackage{textcomp}
\usepackage{mathrsfs}
\usepackage{epstopdf}

\newcommand{\openone}{\leavevmode\hbox{\small1\normalsize\kern-.33em1}}

\catcode`~=11 \def\UrlSpecials{\do\~{\kern -.15em\lower .7ex\hbox{~}\kern .04em}} \catcode`~=13 

\allowdisplaybreaks[4]

\newcommand{\nn}{\nonumber}

\newcommand{\calA}{\mathcal{A}}
\newcommand{\calB}{\mathcal{B}}

\newcommand{\calD}{\mathcal{D}}
\newcommand{\calE}{\mathcal{E}}
\newcommand{\calF}{\mathcal{F}}

\newcommand{\calN}{\mathcal{N}}

\newcommand{\calT}{\mathcal{T}}

\newcommand{\bJ}{\mathbf{J}}

\newcommand{\bs}{\mathbf{s}}
\newcommand{\bS}{\mathbf{S}}

\newcommand{\bu}{\mathbf{u}}
\newcommand{\bU}{\mathbf{U}}

\newcommand{\bV}{\mathbf{V}}

\newcommand{\bx}{\mathbf{x}}
\newcommand{\bX}{\mathbf{X}}

\newcommand{\rmc}{\mathrm{c}}
\newcommand{\rmC}{\mathrm{C}}
\newcommand{\rmd}{\mathrm{d}}

\newcommand{\rme}{\mathrm{e}}

\newcommand{\rmP}{\mathrm{P}}

\newcommand{\rmQ}{\mathrm{Q}}

\newcommand{\rmR}{\mathrm{R}}
\newcommand{\rms}{\mathrm{s}}

\newcommand{\rmT}{\mathrm{T}}

\newcommand{\rmV}{\mathrm{V}}

\newcommand{\bbE}{\mathsf{E}}

\newcommand{\bbN}{\mathbb{N}}

\newcommand{\bbR}{\mathbb{R}}

\DeclareMathAlphabet{\mathbsf}{OT1}{cmss}{bx}{n}
\DeclareMathAlphabet{\mathssf}{OT1}{cmss}{m}{sl}% slanted sans serif

\DeclareSymbolFont{bsfletters}{OT1}{cmss}{bx}{n}  
\DeclareSymbolFont{ssfletters}{OT1}{cmss}{m}{n}
\DeclareMathSymbol{\bsfGamma}{0}{bsfletters}{'000}
\DeclareMathSymbol{\ssfGamma}{0}{ssfletters}{'000}
\DeclareMathSymbol{\bsfDelta}{0}{bsfletters}{'001}
\DeclareMathSymbol{\ssfDelta}{0}{ssfletters}{'001}
\DeclareMathSymbol{\bsfTheta}{0}{bsfletters}{'002}
\DeclareMathSymbol{\ssfTheta}{0}{ssfletters}{'002}
\DeclareMathSymbol{\bsfLambda}{0}{bsfletters}{'003}
\DeclareMathSymbol{\ssfLambda}{0}{ssfletters}{'003}
\DeclareMathSymbol{\bsfXi}{0}{bsfletters}{'004}
\DeclareMathSymbol{\ssfXi}{0}{ssfletters}{'004}
\DeclareMathSymbol{\bsfPi}{0}{bsfletters}{'005}
\DeclareMathSymbol{\ssfPi}{0}{ssfletters}{'005}
\DeclareMathSymbol{\bsfSigma}{0}{bsfletters}{'006}
\DeclareMathSymbol{\ssfSigma}{0}{ssfletters}{'006}
\DeclareMathSymbol{\bsfUpsilon}{0}{bsfletters}{'007}
\DeclareMathSymbol{\ssfUpsilon}{0}{ssfletters}{'007}
\DeclareMathSymbol{\bsfPhi}{0}{bsfletters}{'010}
\DeclareMathSymbol{\ssfPhi}{0}{ssfletters}{'010}
\DeclareMathSymbol{\bsfPsi}{0}{bsfletters}{'011}
\DeclareMathSymbol{\ssfPsi}{0}{ssfletters}{'011}
\DeclareMathSymbol{\bsfOmega}{0}{bsfletters}{'012}
\DeclareMathSymbol{\ssfOmega}{0}{ssfletters}{'012}

\newcommand{\hati}{\hat{i}}
\newcommand{\hatI}{\hat{I}}
\newcommand{\tili}{\tilde{i}}

\newcommand{\hatj}{\hat{j}}
\newcommand{\hatJ}{\hat{J}}
\newcommand{\tilj}{\tilde{j}}

\newcommand{\hats}{\hat{s}}
\newcommand{\hatS}{\hat{S}}

\newcommand{\hatX}{\hat{X}}

\newcommand{\tilX}{\tilde{X}}

\newcommand{\bari}{\bar{i}}
\newcommand{\barj}{\bar{j}}

\newcommand{\barx}{\bar{x}}

\newcommand{\barX}{\bar{X}}

\DeclareMathOperator*{\argmax}{arg\,max}
\DeclareMathOperator*{\argmin}{arg\,min}

\newtheorem{theorem}{Theorem} 
\newtheorem{lemma}[theorem]{Lemma}

\newtheorem{definition}{Definition}

\usepackage{cite}
\usepackage{subfigure} %used to insert pictures

\newcommand{\NT}{N}
\allowdisplaybreaks[2]
\newcommand{\blue}[1]{\textcolor{blue}{#1}} 
\begin{document}

\title{{\huge The   Dispersion of Mismatched Joint~Source-Channel Coding for Arbitrary Sources and Additive Channels} }
\author{\IEEEauthorblockN{Lin Zhou, Vincent Y.~F.~Tan and Mehul Motani} \thanks{The authors are with the Department of Electrical and Computer Engineering (ECE), National University of Singapore (Emails: lzhou@u.nus.edu, vtan@nus.edu.sg, motani@nus.edu.sg). Vincent Y.~F.~Tan is also with the Department of Mathematics, National University of Singapore.}}
\maketitle

\begin{abstract}
We consider a joint source channel coding (JSCC) problem in which we desire to transmit an arbitrary memoryless source over an arbitrary additive  channel. We propose a mismatched    coding architecture that consists of Gaussian codebooks for both the source reproduction sequences and channel codewords. The natural nearest neighbor encoder and decoder, however, need to be judiciously modified to obtain the highest communication rates at finite blocklength. In particular, we consider a unequal error protection (UEP) scheme in which all sources are partitioned into disjoint power type classes. We also regularize the nearest neighbor decoder so that an appropriate measure of the size of each power type class  is taken into account in the decoding strategy. For such an architecture, we derive ensemble-tight second-order and moderate deviations results. Our first-order (optimal bandwidth expansion ratio) result generalizes the seminal results by Lapidoth (1996, 1997). The dispersion of our JSCC scheme is a linear combination of the mismatched dispersions for the channel coding saddle-point problem by Scarlett, Tan and Durisi (2017) and the rate-distortion saddle-point problem by the present authors, thus also generalizing these results.
\end{abstract}


\section{Introduction}
In joint source-channel coding~\cite{shannon1959coding}, one seeks to find  a necessary and sufficient  condition such that a source sequence of length $k$ can be reliably transmitted over a channel in $n$ channel uses in the sense that the excess-distortion probability for a given distortion level $D$ vanishes. This condition is captured by the maximum attainable ratio of $k$ and $n$, also known as {\em rate}. For discrete memoryless systems, Shannon~\cite{shannon1959coding} showed that this maximum attainable rate is ${C}/{R(D)}$, where $C$ is the capacity of a discrete memoryless channel (DMC) and $R(D)$ is the rate-distortion function of a discrete memoryless source (DMS). Shannon showed that, surprisingly, a separation scheme is optimal in this first-order fundamental limit sense. That is, {\em separately} designing a reliable lossy data compression system (source code) and data transmission system (channel code) is optimal. Often, for simplicity, one assumes that these codes are tailored to the source and channel statistics. However, in practice, codes that do not depend on the statistics of the source and channel are of paramount importance. Such codes form the central focus of the present work.

We are primarily inspired by two of Lapidoth's seminal works~\cite{lapidoth1996,lapidoth1997}. In \cite{lapidoth1996}, he showed that for a  channel coding system, if the codebook is Gaussian and the decoder is constrained to be a nearest neighbor or minimum Euclidean distance decoder, regardless of the statistics of the additive noise, the maximum coding rate one can attain is the Gaussian capacity function. This constitutes a {\em robust} communication system because the rate that one attains is at least as good  (i.e., large) as if the noise is Gaussian as long as the code is so designed. In \cite{lapidoth1997}, Lapidoth considered the rate-distortion counterpart of the same problem and showed that the minimum compression rate one can attain for an arbitrary source  is the Gaussian rate-distortion function  if one  uses minimum Euclidean distance encoding and the codebook is Gaussian. Note that for both the source and channel coding systems, the codes are incognizant of the source and channel laws. These problems are also respectively termed as {\em saddle-point} problems because they characterize the extremal input distribution-noise pair (for channel coding) and the source-test channel pair (for source coding). 

We extend these two works of Lapidoth~\cite{lapidoth1996,lapidoth1997} in two distinct directions. First, we consider a joint source-channel coding (JSCC) setup.  In our JSCC scheme, analogously to~\cite{lapidoth1996,lapidoth1997},  one is constrained to use two random Gaussian codebooks, one for the reproduced source sequences  and one for the channel codewords. However, both minimum Euclidean distance encoding and decoding schemes need to be judiciously modified  to ensure that the best (highest) rates are attained. We describe these modifications in greater detail in Section~\ref{sec:main_contri}. We refer  to the encoding and decoding schemes as   modified minimum distance and modified nearest neighbor schemes respectively. The joint scheme is termed the NN-JSCC scheme (NN stands for ``nearest neighbor'').  Second, instead of focusing solely on the first-order asymptotics (capacity and rate-distortion function), we examine the fundamental limits of such a mismatched decoding setup via a more refined lens. Specifically, we study the second-order and moderate deviation asymptotics of the problem. Our results recover the classical results by Lapidoth~\cite{lapidoth1996,lapidoth1997} and more recent works on second-order asymptotics for the saddle-point problems for channel and source coding studied by Scarlett, Tan and Durisi~\cite{scarlett2017mismatch} and the present authors~\cite{zhou2017refined}.

\subsection{Main Contributions and Related Works} \label{sec:main_contri}
Our main contributions are summarized as follows:

\begin{enumerate}
\item We propose a JSCC architecture using Gaussian codebooks, with modified minimum distance encoding and decoding, to transmit an arbitrary memoryless source over an arbitrary additive memoryless channel. \blue{We argue in Section \ref{sec:remark} that this architecture generalizes and unifies works by Lapidoth~\cite{lapidoth1996,lapidoth1997}.} While the Gaussian codebooks are similar to those in~\cite{lapidoth1996,lapidoth1997}, our encoding  and decoding schemes differ somewhat. To capture the JSCC nature of the problem, we draw inspiration from works by Csisz\'ar~\cite{csiszar1980joint} and Wang, Ingber and Kochman~\cite{wang2011dispersion} who respectively established the error exponent and second-order asymptotics for sending a DMS over a DMC. The authors employed the method of types and an unequal error protection (UEP) scheme (cf.\ Shkel, Tan and Draper~\cite{shkel2015unequal}). In our work, we introduce a natural partition of the source sequences into types; however, the notion of types has to be defined carefully since the source need not be discrete. We also regularize  the nearest neighbor decoder~\cite{lapidoth1996} so that an appropriate measure of the size of each type class  is carefully taken into account in the decoding strategy. Our architecture (which is shown in Figure~\ref{systemmodel}) and subsequent analyses allow us to show that the maximum attainable rate is the ratio between the Gaussian capacity and Gaussian rate-distortion function. %Our results thus generalize  Lapidoth's~\cite{lapidoth1996,lapidoth1997}.

\item The main contribution, however, is the derivation of  ensemble-tight second-order coding rates and moderate deviations constants for the architecture so described. By allowing a non-vanishing ensemble excess-distortion probability, we shed light on the backoff from the maximum attainable rate at finite blocklengths. This complements the results of  Kostina and Verd\'u~\cite{kostinajscc} who also derived the dispersion of transmitting a Gaussian memoryless source  (GMS) over an additive white Gaussian noise (AWGN) channel. We show that the mismatched dispersion for our NN-JSCC scheme is a linear combination of the mismatched dispersions in the channel coding saddle-point problem by Scarlett, Tan and Durisi~\cite{scarlett2017mismatch} and the rate-distortion saddle-point problem by the present authors~\cite{zhou2017refined}. For these refined results, there are some intricacies pertaining to what one means by {\em Gaussian codebook}. We consider spherical and i.i.d.\ Gaussian codebooks for both the source reproduction sequences and channel codewords and discuss some subtleties of the  second-order results. 

\item Finally, for both the second-order and moderate deviations asymptotic regimes, we show that the separate source-channel coding scheme by combining the corresponding refined asymptotic results in \cite{scarlett2017mismatch} and \cite{zhou2017refined} for channel-coding and rate-distortion saddle-point problems~\cite{lapidoth1996,lapidoth1997} is strictly sub-optimal compared to the newly proposed NN-JSCC scheme. By combining Lapidoth's results in~\cite{lapidoth1996,lapidoth1997} it is, however, easy to see that separation is first-order optimal.
\end{enumerate}

\subsection{Organization of the Rest of the Paper}
The rest of the paper is organized as follows. In Section \ref{sec:model}, we set up the notation, present our joint source-channel coding system and formulate our problems explicitly. In Section \ref{sec:mainresults}, we present our main results and provide corresponding remarks. The proofs of each of the asymptotic results (second-order and moderate deviations) are provided in Sections \ref{proof:jscc:second} and \ref{proof:jscc:mdc} respectively. Technical results that are not central to  the main exposition are relegated to the Appendices.

\section{The Joint Source-Channel Coding Setup}
\label{sec:model}
\subsection{Notation}
\label{sec:notation}
Random variables and their realizations are in upper (e.g.,\ $X$) and lower case (e.g.,\ $x$) respectively. All sets are denoted in calligraphic font (e.g.,\ $\mathcal{X}$). For any two natural numbers $a$ and $b$ we use $[a:b]$ to denote the set of all natural numbers between $a$ and $b$ (inclusive). We let $\exp\{x\}=e^x$. All logarithms are with respect to base $e$. We use $\rmQ(\cdot)$ to denote the   Gaussian complementary cumulative distribution function (cdf) and $\rmQ^{-1}(\cdot )$ its inverse. Let $X^n:=(X_1,\ldots,X_n)$ be a random vector of length $n$ and $x^n=(x_1,\ldots,x_n)$ be a  realization. We use $\|x^n\|=\sqrt{\sum_i x_i^2}$ to denote the $\ell_2$ norm of a vector $x^n\in\bbR^n$. Given two vectors $x^n$ and $y^n$, the (normalized) quadratic distortion measure is defined as $d(x^n,y^n):=\frac{1}{n}\|x^n-y^n\|^2$.  For any random variable $X$, we use $\Lambda_X(\lambda)$ to denote the cumulant generating function $\lambda\in\bbR\mapsto\log \bbE[\exp \{\lambda X\}]$. For any two sequences $\{a_n\}_{n\geq 1}$ and $\{b_n\}_{n\geq 1}$, we write $a_n\sim b_n$ to mean $\lim_{n\to\infty}{a_n}/{b_n}=1$. We use standard asymptotic notations such as $O(\cdot)$, $o(\cdot)$ and~$\Theta(\cdot)$.

\subsection{System Model}

Consider an arbitrary source  $S$ with probability mass function (PMF) or probability density function (PDF)  $f_S$ satisfying
\begin{align}
\bbE [S^2]=\sigma^2,~\zeta_{\rms}:=\bbE [S^4]<\infty,~\bbE[S^6]<\infty\label{sourceconstraint}.
\end{align} 
Next, consider an arbitrary noise random variable $Z$ with distribution (PMF or PDF) $f_Z$ such that
\begin{align}
\bbE [Z^2]=1,~\zeta_{\rmc}:=\bbE [Z^4]<\infty,~\bbE[Z^6]<\infty\label{channelconstraint}.
\end{align} 
We are interested in using a fixed code to transmit an arbitrary memoryless source $S^k\sim f_S$ to within distortion $D\in(0,\sigma^2)$ over an additive   channel $Y^n=X^n+Z^n$. Here, $X^n$ is the channel input, $Z^n$ is the noise generated i.i.d.\ according to $f_Z$ and $Y^n$ is the corresponding channel output.

To describe our NN-JSCC scheme, we resort to a framework that is ubiquitous in joint source-channel coding, e.g.,~\cite{csiszar1980joint,wang2011dispersion}. We define the notion of {\em power types} for positive reals similar to\cite{zhou2016second}. Let $\xi$ be a positive number. This parameter determines (half) the {\em quantization range}. Furthermore, let the number of source  {\em power type}  (or simply {\em type})  {\em classes} be
\begin{align}
\NT:=\big\lceil 2k\xi\big\rceil\label{def:Ntypes}.
\end{align}
Define the {\em lower limit for the power level} to be
\begin{align}
\Upsilon(0)&:=(1-\xi)\sigma^2\label{def:L0}.
\end{align} 
Given each $i\in[1:\NT]$, define the {\em  type quantization level} and the  {\em power type class} respectively as 
\begin{align}
\Upsilon(i)&:=\Big(1-\xi+\frac{i}{k}\Big)\sigma^2\label{def:Lambdai},\\
\calT_i&:=\Big\{s^k:\Upsilon(i-1)\leq \frac{\|s^k\|^2}{k}<\Upsilon(i)\Big\}\label{def:type}.
\end{align}
Thus, in effect, we are partitioning all length-$k$ source sequences into $\NT$ disjoint subsets $\calT_i, i\in [1:\NT]$ depending on their powers $\|s^k\|^2/k$.  The {\em upper limit for the power level} is $\Upsilon(N)\approx (1+\xi)\sigma^2$ when $k$ is large.  We say that $i \in [1:\NT]$ is the {\em type} or {\em power type} of $s^k$ if $s^k\in\calT_i$. 
Let $\{M_i\}_{i\in[1:\NT]}$ be a set of integers to be specified later. Finally, let
\begin{align}
\calD&:=\{(r,s)\in\bbN^2:r\in[1:\NT],~s\in[1:M_{r}]\}\label{def:cald}.
\end{align}
be a set of   pairs in which the first coordinate denotes the  type and the second coordinate denotes the index of the codeword in a sub-codebook corresponding to that type.

Our NN-JSCC scheme is illustrated in Figure \ref{systemmodel} and defined formally as follows.
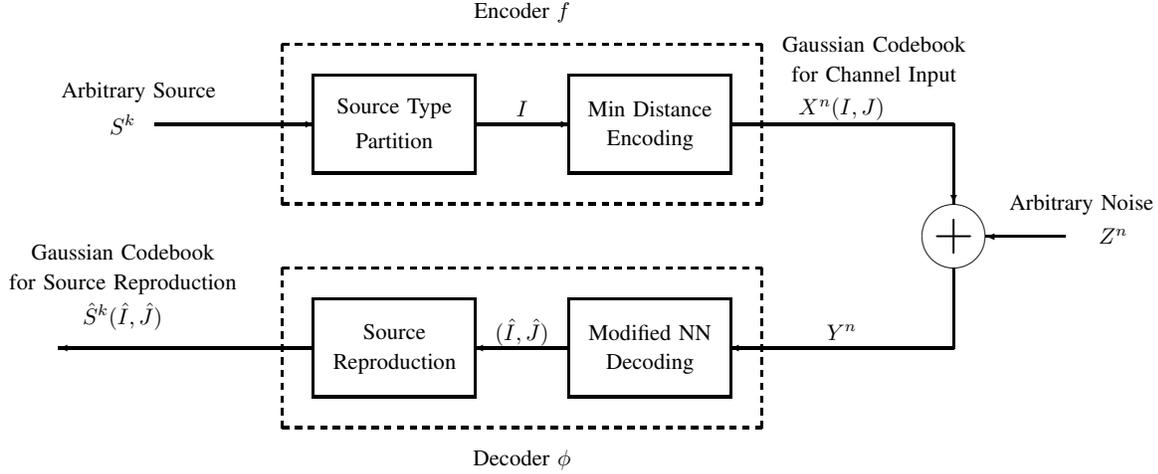
\begin{figure}
\centering
\setlength{\unitlength}{0.5cm}
\scalebox{0.85}{
\begin{picture}(33,15)
\linethickness{1pt}
\put(1.5,11.5){\makebox(0,0){Arbitrary Source}}
\put(1,10.5){\makebox(0,0){$S^k$}}
\put(2,10.5){\vector(1,0){5}}
\put(13.5,14){\makebox(0,0){Encoder $f$}}
\put(6,8){\dashbox{0.2}(15,5)}
\put(7,9){\framebox(5,3)}
\put(9.5,11){\makebox(0,0){Source Type}}
\put(9.5,10){\makebox(0,0){Partition}}
\put(12,10.5){\vector(1,0){3}}
\put(13.5,11){\makebox(0,0){$I$}}
\put(15,9){\framebox(5,3)}
\put(17.5,11){\makebox(0,0){Min Distance}}
\put(17.5,10){\makebox(0,0){Encoding}}
\put(20,10.5){\line(1,0){7}}
\put(23.5,11){\makebox(0,0){$X^n(I,J)$}}
\put(24.5,13){\makebox(0,0){Gaussian Codebook}}
\put(24.5,12){\makebox(0,0){for Channel Input}}
\put(27,10.5){\vector(0,-1){2.5}}
\put(27,7){\circle{2}}
\put(27,7){\makebox(0,0){\Huge$+$}}
\put(27,6){\line(0,-1){2.5}}
\put(30.5,7){\vector(-1,0){2.5}}
\put(32,7){\makebox(0,0){$Z^n$}}
\put(31,8){\makebox(0,0){Arbitrary Noise}}
\put(27,3.5){\vector(-1,0){7}}
\put(23.5,4){\makebox(0,0){$Y^n$}}
\put(15,2){\framebox(5,3)}
\put(17.5,3){\makebox(0,0){Decoding}}
\put(17.5,4){\makebox(0,0){Modified NN}}
\put(15,3.5){\vector(-1,0){3}}
\put(13.5,4){\makebox(0,0){$(\hatI,\hatJ)$}}
\put(7,2){\framebox(5,3)}
\put(9.5,3){\makebox(0,0){Reproduction}}
\put(9.5,4){\makebox(0,0){Source}}
\put(7,3.5){\vector(-1,0){8}}
\put(1,4.5){\makebox(0,0){$\hatS^k(\hatI,\hatJ)$}}
\put(1,5.5){\makebox(0,0){for Source Reproduction}}
\put(1,6.5){\makebox(0,0){Gaussian Codebook}}
\put(13.5,0){\makebox(0,0){Decoder $\phi$}}
\put(6,1){\dashbox{0.2}(15,5)}
\end{picture}
}
\caption{System model of our NN-JSCC scheme which consists of an encoder  $f$ which uses  modified minimum distance encoding  in~\eqref{def:sencoder} and a decoder $\phi$ which uses modified nearest neighbor (NN) decoding in~\eqref{decoder}.}
\label{systemmodel}
\end{figure}
\begin{definition}
\label{def:code}
An $(k,n)$-code for NN-JSCC scheme consists of 
\begin{enumerate}
\item A set of $M_i$ source codewords $\{\hatS^k(i,j)\}_{j=1}^{M_i}$ and a  set of $M_i$ channel codewords $\{X^n(i,j)\}_{j=1}^{M_i}$ for each $i\in[1:\NT]$. The realizations of $\{\hatS^k(i,j)\}_{j=1}^{M_i}$ and  $\{X^n(i,j)\}_{j=1}^{M_i}$ for each $i\in[1:\NT]$  are known to both the encoder and decoder.
\item An encoder $f$ which declares an error if $S^k\notin\bigcup_{i=1}^{\NT}\calT_i$ and uses the following modified minimum distance encoding rule otherwise. The encoder $f$ maps the source sequence $S^k$ into the channel codeword $X^n(I,J)$ if $S^k\in\calT_I$ and $\hatS^k(I,J)$ minimizes the Euclidean distance  over all source codewords in the set $\{\hatS^k(I,\barj)\}_{\barj\in[1:M_I]}$, i.e.,
\begin{align}
J&:=\argmin_{\barj\in[1:M_I]} \|S^k-\hatS^k(I,\barj)\|^2\label{def:sencoder}.
\end{align}
\item A decoder $\phi$ which employs the modified nearest neighbor decoder rule; it declares that the reproduced source sequence is $\hatS^k(\hatI,\hatJ)$ if
\begin{align}
(\hatI,\hatJ):=\argmin_{\substack{(\tili,\tilj)\in\calD}} \|X^n(\tili,\tilj)-Y^n\|^2+2\log M_{\tili}\label{decoder}.
\end{align}
\end{enumerate}
\end{definition}

Throughout the paper, we consider random Gaussian codebooks for both source and channel codebooks for part (i) of Definition~\ref{def:code}. To be specific, we consider the following two types of Gaussian codebooks.
\begin{enumerate}
\item First, we consider   {\em spherical codebooks} where each source codeword $\hatS^k$ (or channel codeword $X^n$) is generated independently and uniformly over a sphere with radius $\sqrt{k(\sigma^2-D)}$ (or $\sqrt{nP}$ where $P$ is a positive number), i.e.,
\begin{align}
\hatS^k\sim f_{\hatS^k}^{\rm{sp}}(\hats^k)&:=\frac{\delta(\|\hats^k\|^2-k(\sigma^2-D))}{A_k(\sqrt{k(\sigma^2-D)})}\label{sourcespcodebook},\\
X^n\sim f_{X^n}^{\rm{sp}}(x^n)&:=\frac{\delta(\|x^n\|^2-nP)}{A_n(\sqrt{nP})}\label{channelspcodebook},
\end{align}
where $\delta(\cdot)$ is the Dirac delta function, $A_k(r) := {k\pi^{k/2}} r^{k-1} / \Gamma(\frac{k+2}{2})$ is the surface area of an $k$-dimensional sphere with radius $r$, and $\Gamma(\cdot)$ is the Gamma function. 
\item Second, we consider   {\em i.i.d.\ Gaussian codebooks} where each source codeword $\hatS^k$ (or channel codeword $X^n$) is generated independently according to a product of  univariate Gaussian distributions each with variance $\sigma^2-D$ (or $P$), i.e.,
\begin{align}
\hatS^k\sim f_{\hatS^k}^{\rm{iid}}(\hats^k)&:=\prod_{i=1}^k\frac{1}{\sqrt{2\pi(\sigma^2-D)}}\exp\bigg\{-\frac{\hats_i^2}{2(\sigma^2-D)}\bigg\}\label{sourceiidcodebook},\\
X^n\sim f_{X^n}^{\rm{iid}}(x^n)&:=\prod_{i=1}^n\frac{1}{\sqrt{2\pi P}}\exp\bigg\{-\frac{x_i^2}{2P}\bigg\}\label{channeliidcodebook}.
\end{align}
\end{enumerate}
For later use, we define the Gaussian capacity and rate-distortion functions as follows:
\begin{align}
\rmC(P)&:=\frac{1}{2}\log(1+P),\\
\rmR(\sigma^2,D)&:=\max\left\{\frac{1}{2}\log\frac{\sigma^2}{D},0\right\}.
\end{align}
Furthermore, define the {\em optimal bandwidth expansion ratio/factor}
\begin{align}
\rho^*(P,\sigma^2,D):=\frac{\rmC(P)}{\rmR(\sigma^2,D)}\label{def:mu},
\end{align}

In other words, the proposed NN-JSCC scheme in Definition~\ref{def:code} consists of a concatenation of a  source code  and a channel code (cf.~\cite[Definition~8]{kostinajscc}). Specifically, the encoder $f$ can be regarded as the concatenation of a source encoder and a channel encoder. The source encoder selects the index $I$ according to source power type class and then selects the sub-index $J$ based on the modified minimum distance encoding rule. The channel encoder maps the output of the source encoder into a channel codeword with index $(I,J)$. The decoder $\phi$ can be regarded as the concatenation of a channel decoder which adopts  the modified nearest neighbor decoding rule to produce  $(\hatI,\hatJ)$ and a source   decoder which declares the source reproduction sequence as the source codeword $\hatS^k(\hatI,\hatJ)$ with this pair of indices.  

\subsection{Motivation for and Remarks on the System Model}
\label{sec:remark}

\blue{
Our motivation for considering the NN-JSCC architecture is, in part, to generalize and unify Lapidoth's works in~\cite{lapidoth1996,lapidoth1997} and, in part,  to obtain the best second-order coding rates for the JSCC problem. Similar to~\cite{lapidoth1996,lapidoth1997}, ours is a \emph{mismatched} coding scheme since  neither the encoder nor the decoder is designed to be optimal with respect to the source and channel. Rather, its design does {\em not} depend on the source and channel statistics.  Hence,  unless the source and channel are {\em Gaussian},  there is {\em mismatch} in the problem. Our NN-JSCC scheme is a UEP-inspired extension of the mismatched coding schemes in the rate-distortion~\cite{lapidoth1997} and channel coding \cite{lapidoth1996} saddle-point problems to the JSCC setting. In fact, if one chooses the parameters so that there is only $\NT=1$ type class  $\calT_1$ (so all the source sequences lie in $\calT_1$), our NN-JSCC scheme degenerates  to a {\em separate} source-channel coding scheme.   For this extreme case, choosing $M$ such that $\log M=n\rmR(\sigma^2,D)+o(n)$ and combining the results in~\cite{lapidoth1996,lapidoth1997}, one concludes that the bandwidth expansion ratio  (ratio of source symbols to channel uses) $\rho^*(P,\sigma^2,D)$ is achievable when the source codebook is a spherical codebook and the channel codebook is either a spherical or i.i.d.\ Gaussian codebook. However, this na\"ive choice results in strictly suboptimal second-order and moderate deviation constants.  For our second-order and moderate deviations results, we exploit the UEP framework of the coding scheme in Fig.~\ref{systemmodel} and choose $\xi$ and $\{M_i\}_{i\in[1:\NT]}$ in a more refined fashion. 
}

%\blue{However, if one na\"ively combines the mismatched codes in \cite{lapidoth1996,lapidoth1997} to form a {\em separate} source-channel code, then one would, in general, obtain \emph{sub-optimal} second-order constants compared to our NN-JSCC scheme. Hence, inspired by the  JSCC schemes in \cite{csiszar1980joint,wang2011dispersion} which were based on unequal error protection (UEP) frameworks, we choose the parameters (e.g., the quantization interval $\xi$, the number of power type classes $N$) judiciously to attain ensemble-tight second-order constants. }

\blue{The complexity (hence practicality or impracticality) of our NN-JSCC coding scheme is almost the same as the schemes in~\cite{lapidoth1996,lapidoth1997}. To wit, we note that both NN encoding and decoding require exponential-time searches over the source and channel codewords. Our scheme incurs an additional search for the index of the power type class that the source $S^k$ lies in; see point~(ii) of Definition~\ref{def:code}.  We design $\xi$ such that the number of type classes is polynomial; the complexity of this search is thus  negligible compared to the aforementioned exponential-time searches. Thus, the ``practicality'' of the proposed scheme is not too dissimilar compared to~\cite{lapidoth1996,lapidoth1997}. } 

\blue{
Despite the fact that the coding scheme is relatively simple and the complexity is almost equal to that in Lapidoth's works~\cite{lapidoth1996,lapidoth1997}, it  remains {\em robust} in the sense the bandwidth expansion ratio $\rho^*(P,\sigma^2,D)$ (which is optimal for  the Gaussian version of the problem) is attained. However, this not necessarily optimal for the  given arbitrary source and arbitrary additive channel. Nonetheless,  the second-order terms can be shown to be ensemble-tight. 
}
% Furthermore, our NN-JSCC scheme is also an extension of the JSCC scheme for transmitting a DMS over a DMC by Csisz\'ar~\cite{csiszar1980joint} who leveraged the idea of unequal error protection to derive the error exponent. In particular, the sources were partitioned into type classes and each type class is then regarded as a message set (in which the messages are distributed uniformly) for which standard channel coding theorems in the error exponents regime apply.

%Third, for our NN-JSCC scheme, careful choices of the parameters $\xi$ and $\{M_i\}_{i\in[1:\NT]}$ are essential. If one chooses $\xi={1}/{(2k)}$  and thus $\NT=1$, then our NN-JSCC scheme degenerates  to a {\em separate} source-channel coding scheme. This scheme is a  combination of the mismatched coding schemes for the rate-distortion and channel coding saddle-point problems in Lapidoth's  celebrated works~\cite{lapidoth1996,lapidoth1997}. For this extreme case, choosing $M$ such that $\log M=n\rmR(\sigma^2,D)+o(n)$ and combining the results in~\cite{lapidoth1996,lapidoth1997}, one  concludes that the rate  (ratio of source symbols to channel uses) ${\rmC(P)}/{\rmR(\sigma^2,D)}$ is achievable when the source codebook is a spherical codebook and the channel codebook is either a spherical or i.i.d.\ Gaussian codebook. For our second-order and moderate deviations results, we choose $\xi$ and $\{M_i\}_{i\in[1:\NT]}$ in a more refined fashion. 

\subsection{Definitions}

Based on the coding scheme in Definition~\ref{def:code}, we see that the {\em (ensemble) excess-distortion probability} is 
\begin{align}
\rmP_{\rme,k,n}
&:=\Pr\{d(S^k,\phi(f(S^k)))>D\}\\
&=\Pr\Big\{S^k\notin\bigcup_{i=1}^{\NT}\calT_i\Big\}+\sum_{i=1}^{\NT}\Pr\Big\{S^k\in\calT_i,~d(S^k,\phi(f(S^k)))>D\Big\}\label{def:excessdp}.
\end{align}
Note that the ensemble excess-distortion probability in \eqref{def:excessdp} is averaged not only over the source and noise distributions, but also over the source and channel codebooks. This is similar to \cite{lapidoth1996,lapidoth1997} which allows us to obtain ensemble-tight results in the spirit of \cite{gallager_ensemble,scarlett2017mismatch,zhou2017refined}.

\blue{For subsequent analyses, let $k^*_{\rm{sp,sp}}(n,\varepsilon,P,\sigma^2,D)$ be the maximal number of source symbols that can be transmitted over the additive noise channel in $n$ channel uses so that the ensemble excess-distortion probability with respect to distortion level $D$ is no larger than $\varepsilon\in(0,1)$ when a spherical codebook is used as both source and channel codebooks. In a similar manner, we can define $k^*_{\rm{sp,iid}}(n,\varepsilon,P,\sigma^2,D)$, $k^*_{\rm{iid,sp}}(n,\varepsilon,P,\sigma^2,D)$ and $k^*_{\rm{iid,iid}}(n,\varepsilon,P,\sigma^2,D)$.\footnote{Throughout the paper, when we use double subscripts consisting of elements of the set $\rm\{sp,iid\}$, the first subscript denotes the nature  of the source codebook (spherical or i.i.d.) and the second denotes the nature of the channel codebook.}}

\begin{definition}
\label{def:secondorder}
Fix any $\varepsilon\in [0,1)$. The {\em spherical-spherical second-order coding rate} is defined as
\begin{align}
L^*_{\rm{sp,sp}}(\varepsilon)
&:=\limsup_{n\to\infty}\frac{1}{\sqrt{n}}\big(n\rho^*(P,\sigma^2,D)-k^*_{\rm{sp,sp}}(n,\varepsilon,P,\sigma^2,D)\big).\label{eqn:vareps}
\end{align}
Similarly, we can define $L^*_{\rm{sp,iid}}(\varepsilon)$, $L^*_{\rm{iid,sp}}(\varepsilon)$ and $L^*_{\rm{iid,iid}}(\varepsilon)$.
\end{definition}

\begin{definition} \label{def:md}
A sequence $\{\eta_n\}_{n\in\bbN}$ is said to be a {\em moderate deviations sequence}\footnote{Our definition of moderate deviations sequence in~\eqref{mdc:constaint} is different from the standard one in for example~\cite{altugwagner2014,polyanskiy2010channel} in which the term $\sqrt{ {n}/{\log n}}$ is replaced by the less stringent $\sqrt{n}$. We require the additional $\sqrt{\log n}$ for technical reasons but it is not restrictive as all sequences of the form $n^{-t}$ for $t\in (0,1/2)$ are, by definition, moderate deviations sequences.  }  if
\begin{align} 
\eta_n\to 0\quad \mbox{and}\quad \sqrt{\frac{n}{\log n}}\eta_n\to\infty\quad \mbox{as}\quad n\to\infty.\label{mdc:constaint}
\end{align}
Let  the length of the source sequence be 
\begin{equation}
k_n:=\lfloor n(\rho^*(P,\sigma^2,D)-\eta_n) \rfloor.
\end{equation}
The {\em spherical-spherical moderate deviations constant} is defined as
\begin{align}
\nu^*_{\rm{sp,sp}}
&:=\liminf_{n\to\infty}-\frac{1}{n\eta_n^2}\log \rmP_{\rme,k_n,n}.
\end{align}
Similarly, we can define  $\nu^*_{\rm{sp,iid}}$, $\nu^*_{\rm{iid,sp}}$ and $\nu^*_{\rm{iid,iid}}$.
\end{definition}

\section{Main Results and Discussions}
\label{sec:mainresults}
\subsection{Preliminaries}
In this subsection, we present some preliminary definitions to be used in presenting our main results.

For $\dagger\in\rm\{sp,iid\}$ and any source sequence $s^k$, note by spherical symmetry that the non-excess-distortion probability $\Pr\{d(s^k,\hatS^k)\leq D\}$, where $\hatS^k\sim f_{\hatS^k}^{\dagger}$, depends on $s^k$ only through its norm $\|s^k\|$. Thus, for any $s^k$ such that $ {\|s^k\|^2}/{k}=p$, we   define
\begin{align}
\Psi_{\dagger}(k,p)&:=\Pr\{d(s^k,\hatS^k)\leq D\},\quad\mbox{where} \quad  \hatS^k\sim f_{\hatS^k}^{\dagger}\label{def:Psikz}.
\end{align}
For $\dagger\in\rm\{sp,iid\}$, when a $\dagger$ Gaussian codebook is used as the random source codebook, for each $i\in[1:\NT]$, we choose 
\begin{align}
\log M_i&:=-\log \Psi_{\dagger}(k,\Upsilon(i))+\log k\label{choosemi}.
\end{align}
We remark that the choice of $M_i$ for any $i\in[1:\NT]$ is universal because it only depends on the quantization level $\Upsilon(i)$~(see~\eqref{def:Lambdai}), which is fixed {\em a priori}, and the type of source codebook~(see~\eqref{sourcespcodebook} and \eqref{sourceiidcodebook}).  It does not depend on the source codebook realization.

The choice of $\xi$ depends on the specific regime (second-order or moderate deviations) and is thus stated later. We need the following definitions of the {\em  mismatched dispersion functions} in \cite{zhou2017refined,scarlett2017mismatch}:
\begin{align}
\rmV_\rms(\zeta_\rms,\sigma^2)&:=\frac{\zeta_\rms-\sigma^4}{4\sigma^4}=\frac{\mathbb{E}[S^4]-(\mathbb{E}[S^2])^2}{4(\mathbb{E}[S^2])^2}\label{def:sourcedispersion},\\
\rmV_\rmc^{\rm{sp}}(\zeta_\rmc,P)&:=\frac{P^2(\zeta_\rmc-1)+4P}{4(P+1)^2},\label{def:channeldispersionspherical}\\
\rmV_\rmc^{\rm{iid}}(\zeta_\rmc,P)&:=\frac{P^2(\zeta_\rmc+1)+4P}{4(P+1)^2}.\label{def:channeldispersioniid}
\end{align}

To simplify the presentation of our main results, recalling the definition of $\rho^*(P,\sigma^2,D)$ in \eqref{def:mu}, for any $\ddagger\in\rm\{sp,iid\}$, define the {\em joint source-channel mismatched dispersion functions} as
\begin{align}
\rmV_{\ddagger}(\zeta_\rms,\sigma^2,\zeta_\rmc,P)
&:=\frac{\rho^*(P,\sigma^2,D)\rmV_\rms(\zeta_\rms,\sigma^2)+\rmV_\rmc^{\ddagger}(\zeta_\rmc,P)}{(\rmR(\sigma^2,D))^2},\\
&=\frac{\rmC(P)\rmV_\rms(\zeta_\rms,\sigma^2)+\rmR(\sigma^2,D)\rmV_\rmc^{\ddagger}(\zeta_\rmc,P)}{(\rmR(\sigma^2,D))^3}\label{def:rmvdagger}
\end{align}

\subsection{Second-Order Asymptotics}

\begin{theorem}
\label{jscc:second}
Let the quantization range be
\begin{align}
\xi&:=\sqrt{\frac{\log k}{k}}\label{def:xi}.
\end{align}
For any $\varepsilon\in[0,1)$ and any $(\dagger,\ddagger)\in\rm\{sp,iid\}^2$, we have
\begin{align}
L^*_{\dagger,\ddagger}(\varepsilon)&=\sqrt{\rmV_{\ddagger}(\zeta_\rms,\sigma^2,\zeta_\rmc,P)}\rmQ^{-1}(\varepsilon).\label{spchsecond}
\end{align}
\end{theorem}
The proof of Theorem \ref{jscc:second} is given in Section \ref{proof:jscc:second}. A few remarks are in order.

First, given a channel codebook, regardless of the choice of the source codebook, the second-order coding rate remains the same. This is consistent with the result in \cite{zhou2017refined} where the present authors showed that the dispersion for the rate-distortion problem using Gaussian codebooks and minimum Euclidean distance encoding remains the same regardless of the particular choice (spherical or i.i.d.) of the Gaussian codebook. Furthermore, given a source codebook, the second-order coding rates are different and depend on the choice of the channel codebook. This is consistent with the result in \cite{scarlett2017mismatch} where Scarlett, Tan, and Durisi showed that the dispersion for the nearest neighbor  decoding over additive non-Gaussian noise channels depends on the particular choice of the channel codebook (spherical or i.i.d.). In particular, the authors of \cite{scarlett2017mismatch} showed that 
\begin{align}
\rmV_\rmc^{\rm{iid}}(\zeta_\rmc,P)&=\rmV_\rmc^{\rm{sp}}(\zeta_\rmc,P)+\frac{1}{2}\Big(\frac{P}{P+1}\Big)^2.
\end{align}

Second, when we particularize our result  to transmitting a GMS   $f_S=\calN(0,\sigma^2)$ over an AWGN channel  with noise distribution  $f_Z=\calN(0,1)$, we have that $\rmV_\rms(\zeta_\rms,\sigma^2)=\frac{1}{2}$ and $\rmV_\rmc^{\rm{sp}}(\zeta_\rmc,P)=\frac{P(P+2)}{2(P+1)^2}$. Hence, we   recover the achievability part in \cite[Theorem 19]{kostinajscc} where Kostina and Verd\'u provided the optimal second-order coding rate of transmitting a GMS over an AWGN channel using spherical source and channel codebooks. Our result in \eqref{spchsecond} shows that the same second-order coding rate can also be achieved when the source  codebook is an i.i.d.\ Gaussian codebook. 

Third, as a corollary of our results in Theorem \ref{jscc:second}, we conclude that for any $\varepsilon\in(0,1)$, regardless of the choices of source and channel codebooks, using our NN-JSCC scheme (see\ Definition \ref{def:code}), we have
\begin{align}
\lim_{n\to\infty}\frac{k^*_{\dagger,\ddagger}(n,\varepsilon,P,  \sigma^2, D)}{n}=\frac{\rmC(P)}{\rmR(\sigma^2,D)},\quad\forall\, (\dagger,\ddagger)\in\rm\{sp,iid\}^2. \label{firstorder}
\end{align}
This strengthens and generalizes Lapidoth's results in~\cite{lapidoth1996,lapidoth1997}. In particular in \cite{lapidoth1997},  he only considered spherical codebooks.

Finally, when we use a separate source-channel coding scheme by combining the models in \cite{lapidoth1996,lapidoth1997} and the results in \cite[Theorem 1]{scarlett2017mismatch} and \cite[Theorem 1]{zhou2017refined}, we obtain that the second-order coding rate for any $(\dagger,\ddagger)\in\rm\{sp,iid\}^2$ is bounded above as 
\begin{align}
L^*_{\dagger,\ddagger}(\varepsilon)&\leq \min_{(\varepsilon_1,\varepsilon_2):\varepsilon_1+\varepsilon_2\leq \varepsilon}\frac{\sqrt{\rho^*(P,\sigma^2,D)\rmV_\rms(\zeta_\rms,\sigma^2)}}{\rmR(\sigma^2,D)}\rmQ^{-1}(\varepsilon_1)+\frac{\sqrt{\rmV_\rmc^{\ddagger}(\zeta_\rmc,P)}}{\rmR(\sigma^2,D)}\rmQ^{-1}(\varepsilon_2).
\end{align}
Hence, the separate source-channel coding scheme by combining the rate-distortion and channel coding saddle-point setups in~\cite{lapidoth1997,lapidoth1996} is strictly  suboptimal in the second-order sense unless $\rmV_\rms(\zeta_\rms,\sigma^2)$ or $\rmV_\rmc^{\ddagger}(\zeta_\rmc,P)$ is zero.

\subsection{Moderate Deviations}

Before presenting our results, we need the following assumptions on the source and channel parameters.
\begin{enumerate}
\item\label{assump1} $\rmV_\rms(\zeta_\rms,\sigma^2)$ is positive;
\item\label{assump2} The cumulant generating functions $\Lambda_{S^2}(\lambda)$, $\Lambda_{Z^2}(\lambda)$, $\Lambda_{\tilX Z}(\lambda)$ are all finite in a neighborhood around the origin, where $\tilX$ is a Gaussian random variable with zero mean and variance one and it is independent of all other random variables. 
\end{enumerate}
\begin{theorem}
\label{jscc:mdc}
Let the quantization range be
\begin{align}
\xi&:=\eta_n^{3/4}\label{def:ximdc},
\end{align}
where $\eta_n = \omega( \sqrt {n^{-1}\log n})$ denotes the backoff from the first-order fundamental limit $\rmC(P)/\rmR(\sigma^2,D)$ as delineated in Definition \ref{def:md}.
Under conditions \eqref{assump1} and \eqref{assump2} as stated above,  for any $(\dagger,\ddagger)\in\rm\{sp,iid\}^2$, we have
\begin{align}
\nu^*_{\rm{\dagger,\ddagger}}&=\frac{1}{2\rmV_{\ddagger}(\zeta_\rms,\sigma^2,\zeta_\rmc,P)}.
\end{align}
\end{theorem}
The proof of Theorem \ref{jscc:mdc} is given in Section \ref{proof:jscc:mdc}. A few remarks are in order.

First, similarly to the second-order asymptotics in Theorem \ref{jscc:second}, we observe that the dispersion $\rmV_{\ddagger}(\zeta_\rms,\sigma^2,\zeta_\rmc,P)$ plays an important role in the sub-exponential decay of the ensemble excess-distortion probability. Furthermore, the moderate deviations performance only depends on the choice of the channel codebook.

Second, in the proof of Theorem \ref{jscc:mdc}, we need to make use of a moderate deviation theorem for functions of independent but not necessarily identically distributed random vectors   (see\ Lemma~\ref{mdc4funcofirv}). 

Finally, we remark that if one uses a separate source-channel coding scheme by combining the models in \cite{lapidoth1996} and \cite{lapidoth1997}, then under same conditions, the optimal MDC satisfies that for any $(\dagger,\ddagger)\in\rm\{sp,iid\}^2$
\begin{align}
\nu^*_{\rm{\dagger,\ddagger}}&\leq \min\bigg\{\frac{(\rmR(\sigma^2,D))^2}{2\rho^*(P,\sigma^2,D)\rmV_\rms(\zeta_\rms,\sigma^2)},\frac{(\rmR(\sigma^2,D))^2}{2\rmV_\rmc^{\ddagger}(\zeta_\rmc,P)}\bigg\}.
\end{align}
Hence, the separate source-channel coding scheme by combining the rate-distortion and channel coding saddle-point setups in \cite{lapidoth1997,lapidoth1996} is strictly sub-optimal in terms of moderate deviations asymptotics.

\section{Proof of Second-Order Asymptotics (Theorem \ref{jscc:second})}
\label{proof:jscc:second}
To establish Theorem \ref{jscc:second}, we need to prove the results for four combinations of source and channel codebooks where each codebook can either be a spherical or an i.i.d.\ Gaussian codebook. In Section \ref{sec:prelinimary}, we present preliminary results. In Sections~\ref{sec:proofach} and~\ref{sec:proofconversesr}, we present the achievability and converse proofs of Theorem \ref{jscc:second} respectively.

\subsection{Preliminaries} 
\label{sec:prelinimary}
In this subsection, we present some preliminary results for subsequent analyses.

\subsubsection{Analysis of Excess-Distortion Events}
\label{sec:preerrorevents}
Recall our NN-JSCC scheme in Definition \ref{def:code} and Figure \ref{systemmodel}. Given any source sequence $S^k$, the encoder $f$ declares an error if $S^k\notin\bigcup_{i=1}^{\NT}\calT_i$ and maps it into the codeword $X^n(I,J)$ if $S^k\in\calT_I$ and $\hatS^k(I,J)$ minimizes the Euclidean distance with respect to $S^k$ over all codewords in the subcodebook $\{\hatS^k(I,\barj)\}_{\barj\in[1:M_i]}$. Given the channel output $Y^n$, the channel decoder uses the modified nearest neighbor decoding (see~\eqref{decoder}) to find $(\hatI,\hatJ)$ and declares $\hatS^k(\hatI,\hatJ)$ as the reproduced source sequence.

For our NN-JSCC scheme, an excess-distortion event occurs if and only if one of the following events occur:
\begin{enumerate}
\item $S^k\notin\bigcup_{i=1}^{\NT}\calT_i$;
\item $S^k\in\calT_i$ for some $i\in[1:\NT]$ (i.e., $I=i$ for some $i\in[1:\NT]$) and one of the following events occur:
\begin{enumerate}
\item The message pair $(i,J)$ is transmitted correctly and the distortion is greater than $D$, i.e,
\begin{align}
\calE_1(i)&:=
\{(\hatI,\hatJ)=(i,J),~d(S^k,\hatS^k(i,J))>D\}\label{eevent:1},
\end{align}
\item The message $i$ is transmitted incorrectly and the distortion is greater than $D$, i.e.,
\begin{align}
\calE_2(i)&:=
\{\hatI\neq i,~d(S^k,\hatS^k(\hatI,\hatJ))>D\}\label{eevent:2}.
\end{align}
\item The message $i$ is transmitted correctly, the message $J$ is transmitted incorrectly  and the distortion is greater than $D$, i.e.,
\begin{align}
\calE_3(i)&:=
\{\hatI=i,\hatJ\neq J,~d(S^k,\hatS^k(\hatI,\hatJ))>D\}\label{eevent:3}.
\end{align}
\end{enumerate} 
\end{enumerate} 
Using the definition of the ensemble excess-distortion probability in \eqref{def:excessdp} and the definitions of error events in~\eqref{eevent:1},~\eqref{eevent:2} and~\eqref{eevent:3}, we see that
\begin{align}
\rmP_{\rme,k,n}
&=\Pr\Big\{S^k\notin\bigcup_{i=1}^{\NT}\calT_i\Big\}
+\sum_{i=1}^{\NT} \Pr\{S^k\in\calT_i\}\Pr\{\calE_1(i)\cup\calE_2(i)\cup\calE_3(i)|S^k\in\calT_i\}\label{excessp:alt}.
\end{align}

In subsequent analyses for the achievability parts, we upper bound the ensemble excess-distortion probability as follows:
\begin{align}
\rmP_{\rme,k,n}
\nn&\leq \Pr\Big\{S^k\notin\bigcup_{i=1}^{\NT}\calT_i\Big\}
+\sum_{i=1}^{\NT} \int_{s^k\in\calT_i}\Pr\big\{d(s^k,\hatS^k(i,J))>D\big\} f_{S^k}(s^k) \rmd s^k\\
&\qquad+\sum_{i=1}^{\NT}\int_{s^k\in\calT_i} \Pr\{(\hatI,\hatJ)\neq (i,J)\}f_{S^k}(s^k) \rmd s^k\label{uppfirst},
\end{align}
where \eqref{uppfirst} follows by i) using the union bound, ii) ignoring the requirement that the message pair $(i,J)$ is transmitted correctly in $\calE_1(i)$, iii) ignoring the excess-distortion event in $\calE_2(i)$ and $\calE_3(i)$, and iv) noting that 
\begin{align}
\Pr\{\{\hatI\neq i\}\cup\{\hatI=i,~\hatJ\neq J\}\}=\Pr\{(\hatI,\hatJ)\neq (i,J)\}\label{simplefact}.
\end{align}
Note that in the sum in~\eqref{uppfirst}, the first two probabilities are with respect to the joint distribution of the source sequence and source codebook while the last probability is with respect to distributions of the channel codebook and the noise.

In subsequent analyses for converse parts, we lower bound the ensemble excess-distortion probability as
\begin{align}
\rmP_{\rme,k,n}
&\geq \Pr\Big\{S^k\notin\bigcup_{i=1}^{\NT}\calT_i\Big\}+\sum_{i=1}^{\NT}\Pr\{S^k\in\calT_i\}\Pr\{\calE_2(i)\cup\calE_3(i)|S^k\in\calT_i\}\label{lowfirstfirst}\\
\nn&=\Pr\Big\{S^k\notin\bigcup_{i=1}^{\NT}\calT_i\Big\}+\sum_{i=1}^{\NT} \Pr\{S^k\in\calT_i\}\Pr\{\{\hatI\neq i~\}\cup\{\hatI=i,~\hatJ\neq J\}|S^k\in\calT_i\}\\
&\qquad\qquad-\Pr\{S^k\in\calT_i\}\Pr\Big\{\big\{\{\hatI\neq i\}\cup\{\hatI=i,~\hatJ\neq J\}\big\}\cap\{d(S^k,\hatS^k(\hatI,\hatJ))\leq D\} \Big|S^k\in\calT_i\Big\}\label{lowfirst}\\
\nn&\geq \Pr\Big\{S^k\notin\bigcup_{i=1}^{\NT}\calT_i\Big\}+\sum_{i=1}^{\NT} \Pr\{S^k\in\calT_i,~(\hatI,\hatJ)\neq (i,J)\}-\Pr\{S^k\in\calT_i,~\hatI\neq i,~d(S^k,\hatS^k(\hatI,\hatJ))\leq D\}\\
&\qquad\qquad-\Pr\{S^k\in\calT_i,~\hatI=i,~\hatJ\neq J,~d(S^k,\hatS^k(\hatI,\hatJ))\leq D\},\label{lowfirst2}
\end{align}
where \eqref{lowfirst} follows since $\Pr\{\calA\cap\calB\}=\Pr\{\calA\}-\Pr\{\calA\cap\calB^\rmc\}$ for any two sets $\calA$ and $\calB$, and \eqref{lowfirst2} follows by using \eqref{simplefact} and applying the union bound on the final term in \eqref{lowfirst}.

\subsubsection{Analysis of the Output of the Channel Decoder}
\label{sec:channeloutput}
First, we clarify the relationship of the random variables involved in our joint source channel coding ensemble (see Definition \ref{def:code}). In particular, we specify the dependence of the channel output $(\hatI,\hatJ)$ on other random variables such as the source sequence $S^k$ and the source codebook. The results in this subsection hold regardless the choices of source and channel codebooks.

For simplicity, let
\begin{align}
\hat{\bS}_i&:=\{\hatS^k(i,\barj)\}_{\barj\in[1:M_i]},\\
\hat{\bS}&:=\cup_{i=1}^{\NT}\hat{\bS}_i,\\
\bX&:=\cup_{i=1}^{\NT}\{X^n(i,\barj)\}_{\barj\in[1:M_i]},
\end{align}
and let $\hat{\bs}_i$, $\hat{\bs}$ and $\bx$ be the corresponding realizations. Furthermore, for any $i\in[1:\NT]$ and for $\dagger\in\rm\{sp,iid\}$, let
\begin{align}
f_{\hat{\bS}_i}^{\dagger}(\hat{\bs}_i)&:=\prod_{\barj=1}^{M_i}f_{\hatS^k}^{\dagger}(\hats^k(i,\barj)),\\
f_{\hat{\bS}}^{\dagger}(\hat{\bs})&:=\prod_{i=1}^{\NT}f_{\hat{\bS}_i}^{\dagger}(\hat{\bs}_i).
\end{align}

Recall the definition of our NN-JSCC scheme in Definition \ref{def:code} and the definition of $\calD$ in \eqref{def:cald}. For any $i\in[1:\NT]$, given $S^k\in\calT_i$ and $\hat{\bS}_i$ (and thus $J=\argmin_{\barj\in[1:M_i]}\|S^k-\hatS^k(i,\barj)\|^2$ (see \eqref{def:sencoder})), the output of the channel decoder is
\begin{align}
(\hatI,\hatJ)&=\argmin_{(\tili,\tilj)\in\calD} \|X^n(\tili,\tilj)-Y^n\|^2+2\log M_{\tili}\\
&=\argmin_{(\tili,\tilj)\in\calD} \|X^n(\tili,\tilj)-(X^n(i,J)+Z^n)\|^2+2\log M_{\tili}\label{eqn:cdecoder}.
\end{align}
From \eqref{eqn:cdecoder}, we conclude that the output of the channel decoder $(\hatI,\hatJ)$ depends on the source sequence $S^k$ and the source codebook $\bS$ only through the type of the source sequence and the subcodebook $\hat{\bS}_i$, i.e., for any $s^k\in\calT_i$ and any $\hat{\bs}$,
\begin{align}
\Pr\{(\hatI,\hatJ)=(\hati,\hatj)|S^k=s^k,~\hat{\bS}=\hat{\bs}\}
&=\Pr\{(\hatI,\hatJ)=(\hati,\hatj)|S^k=s^k,~\hat{\bS}_i=\hat{\bs}_i\}\label{outputc:depends},
\end{align}
where $\hat{\bs}_i$ is the $i$-th subcodebook of $\hat{\bs}$. Note that the probability in \eqref{outputc:depends} is with respect to the channel codebook.

Given any $(x^n,y^n)$, the {\em mismatched information density} (see~\cite[Eqns.~(28)-(29)]{scarlett2017mismatch}) is defined as
\begin{align}
\imath(x^n;y^n)&:=n\rmC(P)+\frac{\|y^n\|^2}{2(P+1)}-\frac{\|y^n-x^n\|^2}{2}\label{def:mismatchid}.
\end{align}
For any $i\in[1:\NT]$ and any $\ddagger\in\rm\{sp,iid\}$, let
\begin{align}
\overline{h}_{\ddagger}(n,i)&:=\min\Bigg\{1,\mathbb{E}\Bigg[\sum_{\tili=1}^{\NT}M_{\tili}\Pr\bigg\{\imath(\barX^n;Y^n)\leq \imath(X^n(i,j);Y^n)-\log \frac{M_i}{M_{\tili}}\bigg|X^n,Y^n\bigg\}\Bigg]\Bigg\},\label{def:ohni}\\
\underline{h}_{\ddagger}(n,i)&:=1-\Big(1-\mathbb{E}\big[\Pr\big\{\|\barX^n-Y^n\|^2\leq \|X^n-Y^n\|\big|X^n,Y^n\big\}\big]\Big)^{M_i-1}\label{def:uhni},
\end{align}
where in \eqref{def:ohni} and \eqref{def:uhni}, the tuple $(\barX^n,X^n,Y^n)$ is distributed according to the following joint distribution
\begin{align}
f_{X^n}^{\ddagger}(\barx^n)f_{X^n}^{\ddagger}(x^n)f_{Y^n|X^n}(y^n|x^n)\label{def:joint}.
\end{align}
For simplicity, given $s^k\in\calT_i$ and $\hat{\bs}_i$ for any $i\in[1:\NT]$, we let \begin{align}
j(s^k,\hat{\bs}_i):=\argmin_{\barj\in[1:M_i]}d(s^k,\hats^k(i,\barj)) \label{def:jskhsi}.
\end{align} 

In the following lemma, we present bounds on the error probability of the channel decoder conditioned on a source sequence (within a type class) and a subcodebook realization.
\begin{lemma}
\label{channelerror}
For any $i\in[1:\NT]$ and any $\ddagger\in\rm\{sp,iid\}$, given any $s^k\in\calT_i$ and any subcodebook $\hat{\bs}_i$, we have
\begin{align}
\Pr\{(\hatI,\hatJ)\neq (i,j(s^k,\hat{\bs}_i))|S^k=s^k,~\hat{\bS}_i=\hat{\bs}_i\}
&\leq \overline{h}_{\ddagger}(n,i)\label{upp:pijwrong},\\
\Pr\{(\hatI,\hatJ)\neq (i,j(s^k,\hat{\bs}_i))|S^k=s^k,~\hat{\bS}_i=\hat{\bs}_i\}
&\geq \underline{h}_{\ddagger}(n,i)\label{low:pijwrong}.
\end{align} 
\end{lemma}
The proof of Lemma \ref{channelerror}, inspired by and similar to that in~\cite{scarlett2017mismatch} by Scarlett, Tan and Durisi, is available in Appendix \ref{proofchannelerror}. Note that Lemma \ref{channelerror} holds regardless of the choice of the channel codebook. We remark that the upper bound given in \eqref{upp:pijwrong} is an extension of RCU bound in \cite[Theorem 16]{polyanskiy2010finite} to the unequal message protection setting~(see \cite{shkel2015unequal}) and the lower bound in \eqref{low:pijwrong} is a proxy of the RCU bound in the other direction.

\subsubsection{Existing Results for Non-Excess-Distortion Probabilities}
We now recall existing results concerning non-excess-distortion probabilities (see\ \eqref{def:Psikz}) from~\cite{zhou2017refined}. Let $
r_1:=\sqrt{\sigma^2-D}-\sqrt{D}$ and 
$r_2 :=\sqrt{\sigma^2-D}+\sqrt{D}$. Furthermore, given $p\in\bbR_+$ let
\begin{align}
R_{\rm{sp}}(p)&:=-\frac{1}{2}\log \bigg(1-\frac{(p+\sigma^2-2D)^2}{4p(\sigma^2-D)}\bigg)\label{def:rspz},\\
s^*(p)&:=\max\bigg\{0,\frac{\sigma^2-3D+\sqrt{(\sigma^2-D)^2+4pD}}{4D}\bigg\}\label{def:sstar},\\
R_{\rm{iid}}(p)&:=\frac{1}{2}\log(1+2s^*(p))+\frac{s^*(p)p}{(1+2s^*(p))(\sigma^2-D)}-\frac{s^*(p)D}{\sigma^2-D}\label{def:riidsapha}.
\end{align}
Finally, given $p\in\bbR_+$ and $k\in\bbN$, let
\begin{align}
\underline{g}(k,p)&:=\frac{\Gamma(\frac{k+2}{2})}{\sqrt{\pi}k\Gamma(\frac{k+1}{2})}\exp(-(k-1)R_{\rm{sp}}(p)),\label{def:ugkz}\\
\overline{g}(k,p)&:=\frac{\Gamma(\frac{k}{2})}{\Gamma(\frac{k-1}{2})\sqrt{\pi}}\exp(-(k-3)R_{\rm{sp}}(p))\label{def:ogkz}.
\end{align}

Recalling the definition of $\Psi_{\dagger}(\cdot)$ in \eqref{def:Psikz} and the results in \cite{zhou2017refined}, we have the following lemma.
\begin{lemma}
\label{property}
The following claims hold. 
\begin{enumerate}
\item Bounds on $\Psi_{\rm{sp}}(k,p)$:
\begin{enumerate}
\item If $p\leq r_1^2$ or $p\geq r_2^2$, then $\Psi_{\rm{sp}}(k,p)=0$;
\item If $p\in(r_1^2,r_2^2)$, then 
$
\Psi_{\rm{sp}}(k,p)\geq \underline{g}(k,p);
$
\item If $p\in(r_1^2,r_2^2)$ and $p+\sigma^2-2D\geq 0$, then
$
\Psi_{\rm{sp}}(k,p)\leq \overline{g}(k,p).
$
\end{enumerate}
\item Properties of $R_{\rm{sp}}(p)$:
\begin{enumerate}
\item $R_{\rm{sp}}(p)$ is increasing in $p$ if $p\geq |\sigma^2-2D|$;
\item $R_{\rm{sp}}(\sigma^2)=\frac{1}{2}\log\frac{\sigma^2}{D}$.
\end{enumerate}
\item Bounds on $\Psi_{\rm{iid}}(k,p)$: 
\begin{align}
\Psi_{\rm{iid}}(k,p)\sim \frac{\exp(-kR_{\rm{iid}}(p))}{s^*(p)\sqrt{\kappa(s^*(p),p)2\pi k}}\quad\mathrm{as} \quad k\to\infty,
\end{align}
where 
\begin{align}
\kappa(s,p)&:=\frac{((\sigma^2-D)(1+2s)+2p)^2}{(\sigma^2-D)(1+2s)^3}\label{def:lambda2s}.
\end{align}
\item Properties of $R_{\rm{iid}}(p)$:
\begin{enumerate}
\item $R_{\rm{iid}}(p)$ is increasing in $p$ for all $p\geq \max\{0,2D-\sigma^2\}$;
\item $R_{\rm{iid}}(\sigma^2)=\frac{1}{2}\log\frac{\sigma^2}{D}$.
\end{enumerate}
\end{enumerate}
\end{lemma}
Using the definitions in \eqref{def:rspz} and \eqref{def:riidsapha}, we conclude that for any $p\in(r_1^2,r_2^2)$ such that $p\geq 2D-\sigma^2$, $R_{\rm{sp}}(p)$ and $R_{\rm{iid}}(p)$ have the same Taylor expansions at $p=\sigma^2$ which can be written as follows:
\begin{align}
R_{\rm{sp}}(p)=R_{\rm{iid}}(p)= \rmR(\sigma^2,D)+\frac{p-\sigma^2}{2\sigma^2}+O\big( (p-\sigma^2)^2\big) ,\quad \mbox{as}\quad p \to \sigma^2 \label{tltobeused}.
\end{align}

\subsection{Achievability Proof}
\label{sec:proofach}
Fix any $\varepsilon\in(0,1)$ and any $\ddagger\in\rm\{sp,iid\}$, let $k$ be chosen such that
\begin{align}
k=n\rho^*(P,\sigma^2,D)-\sqrt{n\rmV_{\ddagger}(\zeta_\rms,\sigma^2,\zeta_\rmc,P)}\rmQ^{-1}(\varepsilon)\label{kachsecond},
\end{align}
where $\rho^*(P,\sigma^2,D)$ is defined in \eqref{def:mu} and $\rmV_{\ddagger}(\zeta_\rms,\sigma^2,\zeta_\rmc,P)$ is defined in \eqref{def:rmvdagger}. In particular, $k$ is linear in $n$, \blue{i.e., $k=\Theta(n)$.}

Using the definitions of $\NT$ in \eqref{def:Ntypes} and $\xi$ in \eqref{def:xi}, we obtain that the number of type classes is bounded as
\begin{align}
\NT\leq 2k\xi+1=2\sqrt{k\log k}+1\label{nt:second}.
\end{align}
Note that for $k$ large enough, if $s^k\in\calT_i$ (see \eqref{def:type}) for any $i\in[1:\NT]$, we have 
\begin{align}
r_1^2\leq \max\{|\sigma^2-2D|,2D-\sigma^2\}\leq \frac{\|s^k\|^2}{k}\leq r_2^2\label{stypical}.
\end{align}

Recall the definition of $j(s^k,\hat{\bs}_i)$ in \eqref{def:jskhsi}. Using the result in \eqref{uppfirst}, for any $(\dagger,\ddagger)\in\rm\{sp,iid\}^2$, we can upper bound the ensemble excess-distortion probability as follows:
\begin{align}
\rmP_{\rme,k,n}
\nn&\leq \Pr\Big\{S^k\notin\bigcup_{i=1}^{\NT}\calT_i\Big\}
+\sum_{i=1}^{\NT} \int_{s^k\in\calT_i}(1-\Pr_{f_{\hatS^k}^{\dagger}}\{d(s^k,\hatS^k)\leq D\})^{M_i} f_{S^k}(s^k) \rmd s^k\\*
&\qquad+\sum_{i=1}^{\NT}\int_{s^k\in\calT_i}\bigg(\int_{\hat{\bs}_i}\Pr\big\{(\hatI,\hatJ)\neq (i,j(s^k,\hat{\bs}_i))\big|S^k=s^k,~\hat{\bS}_i=\hat{\bs}_i\big\} f_{\hat{\bS}_i}(\hat{\bs}_i)\rmd \hat{\bs}_i\bigg)f_{S^k}(s^k) \rmd s^k\label{useseveral}\\
\nn&\leq \Pr\Big\{S^k\notin\bigcup_{i=1}^{\NT}\calT_i\Big\}
+\sum_{i=1}^{\NT} \int_{s^k\in\calT_i}(1-\Pr_{f_{\hatS^k}^{\dagger}}\{d(s^k,\hatS^k)\leq D\})^{M_i} f_{S^k}(s^k) \rmd s^k\\*
&\qquad+\sum_{i=1}^{\NT}\int_{s^k\in\calT_i} \overline{h}_{\ddagger}(n,i) f_{S^k}(s^k) \rmd s^k\label{ss:step2.00}\\
\nn&\leq \Pr\Big\{S^k\notin\bigcup_{i=1}^{\NT}\calT_i\Big\}
+\sum_{i=1}^{\NT} \int_{s^k\in\calT_i}\exp(-M_i\Psi_{\dagger}(k,\|s^k\|^2/k))f_{S^k}(s^k) \rmd s^k\\
&\qquad+\sum_{i=1}^{\NT}\int_{s^k\in\calT_i} \overline{h}_{\ddagger}(n,i) f_{S^k}(s^k) \rmd s^k\label{ss:step2}\\
&\leq \Pr\Big\{S^k\notin\bigcup_{i=1}^{\NT}\calT_i\Big\}
+\exp(-k)+\sum_{i=1}^{\NT}\Pr \{S^k\in\calT_i \}\overline{h}_{\ddagger}(n,i),\label{ss:step3}
\end{align}
where \eqref{useseveral} follows since $d(s^k,\hatS^k(i,J))>D$ implies that $d(s^k,\hatS^k(i,\barj))>D$ for all $\barj\in[1:M_i]$ and each source codeword $\hatS^k(i,\barj)$ is generated independently (see also \cite[Theorem 9]{kostina2012fixed}), \eqref{ss:step2.00} follows from \eqref{upp:pijwrong}, \eqref{ss:step2} follows from the inequality $(1-a)^M\leq \exp(-Ma)$ for all $a\in[0,1]$ and the definition of $\Psi_{\rm{\dagger}}(\cdot)$ in \eqref{def:Psikz}, and \eqref{ss:step3} follows since i) $\Psi_{\dagger}(k,p)$ is decreasing in $p$ for $p\geq \max\{\sigma^2-2D,|\sigma^2-2D|\}$ (which is implied by conclusions (ii)-a) and iv)-a) in Lemma~\ref{property}), ii) $s^k\in\calT_i$ for some $i\in[1:\NT]$ implies that \eqref{stypical} holds and $\frac{\|s^k\|^2}{k}\leq \Upsilon(i)$ (see \eqref{def:type}), and iii) the choice of $\log M_i$ in \eqref{choosemi}.

We bound the first term in \eqref{ss:step3} by invoking the Berry-Esseen Theorem. Let $\rmV:=\zeta_\rms-\sigma^4$ and $T:=\mathbb{E}[|S^2-\sigma^2|^3]$. Then, we have, 
\begin{align}
\Pr\bigg\{S^k\notin\bigcup_{i=1}^{\NT}\calT_i\bigg\}
&\leq \Pr\Big\{\frac{1}{k}\sum_{i=1}^k (S_i^2-\sigma^2)>\xi\sigma^2\Big\}+\Pr\Big\{\frac{1}{k}\sum_{i=1}^k (S_i^2-\sigma^2)<-\xi\sigma^2\Big\}\label{useti}\\
&\leq 2\rmQ\Big(\sqrt{\frac{\log k}{\rmV}}\sigma^2\Big)+\frac{12T}{\sqrt{k}\rmV^{3/2}}\label{usebetheorem}\\
&\leq 2\exp\Big(-\frac{\sigma^4 \log k }{2\rmV}\Big)+\frac{12T}{\sqrt{k}\rmV^{3/2}}\label{uprmq},
\end{align}
where \eqref{useti} follows since $s^k\notin\bigcup_{i=1}^{\NT}\calT_i$ implies $\frac{\|s^k\|^2}{k}<(1-\xi)\sigma^2$ or $\frac{\|s^k\|^2}{k}>(1+\xi)\sigma^2$, \eqref{usebetheorem} follows from the Berry-Esseen theorem ($T$  is finite as the sixth moment of the source is finite) and definition of $\xi$ in \eqref{def:xi}, and \eqref{uprmq} follows from the bound $\rmQ(a)\leq \exp(-\frac{a^2}{2})$. The upper bound in \eqref{uprmq} tends to $0$ as $k\to\infty$. 

The following lemma is essential to bound the final term in \eqref{ss:step3}.
\begin{lemma}
\label{ach:essential}
For any $(\dagger,\ddagger)\in\rm\{sp,iid\}^2$, we have
\begin{align}
\sum_{i=1}^{\NT}\Pr \{S^k\in\calT_i \}\overline{h}_{\ddagger}(n,i)
&\leq \rmQ\Bigg(\rmQ^{-1}(\varepsilon)+O\bigg(\frac{\log n}{\sqrt{n}}\bigg)\Bigg)+O\bigg(\frac{1}{\sqrt{n}}\bigg)\label{eqn:achessential}.
\end{align}
\end{lemma}
The proof of Lemma \ref{ach:essential} is deferred to the end of this subsection.

For any $(\dagger,\ddagger)\in\rm\{sp,iid\}^2$, combining \eqref{ss:step3}, \eqref{uprmq} and \eqref{eqn:achessential} and using the definitions of $\rho^*(P,\sigma^2,D)$ in \eqref{def:mu} and $\rmV_{\ddagger}(\zeta_\rms,\sigma^2,\zeta_\rmc,P)$ in \eqref{def:rmvdagger}, we conclude that with the choice of $k$ in \eqref{kachsecond}, we have
\begin{align}
\limsup_{n\to\infty}\rmP_{\rme,k,n}\leq \varepsilon.
\end{align}
Therefore, we have shown that for any $\varepsilon\in[0,1)$ and any $(\dagger,\ddagger)\in\rm\{sp,iid\}^2$,
\begin{align}
L_{\dagger,\ddagger}^*\leq \sqrt{\rmV_{\ddagger}(\zeta_\rms,\sigma^2,\zeta_\rmc,P)}\rmQ^{-1}(\varepsilon).
\end{align}

\begin{proof}[Proof of Lemma \ref{ach:essential}]
Using the definitions of the mismatched information density in \eqref{def:mismatchid} and $\overline{h}_{\ddagger}(n,i)$ in \eqref{def:ohni}, for any $\ddagger\in\rm\{sp,iid\}$, we obtain that
\begin{align}
\overline{h}_{\ddagger}(n,i)
&\leq \Pr\bigg\{\sum_{\tili=1}^{\NT}M_{\tili}K_0\exp\bigg(-\imath(X^n(i,j);Y^n)+\log \frac{M_i}{M_{\tili}}\bigg)\geq \frac{1}{\sqrt{n}}\bigg\}+\frac{1}{\sqrt{n}},\label{usescarlett}\\
&=\Pr\bigg\{\NT K_0 M_i\exp(-\imath(X^n(i,j);Y^n))\geq \frac{1}{\sqrt{n}}\bigg\}+\frac{1}{\sqrt{n}}\\
%&=\Pr\bigg\{\frac{(P+1)\|Z^n\|^2-\|X^n(i,j)+Z^n\|^2}{2(P+1)}\geq n\rmC(P)-\log M_i-\log (\NT K_0\sqrt{n})\bigg\}+\frac{1}{\sqrt{n}}\label{usechannel}\\
&=\Pr\bigg\{\frac{(P+1)\|Z^n\|^2-\|X^n+Z^n\|^2}{2(P+1)}\geq n\rmC(P)-\log M_i-\log (\NT K_0\sqrt{n})\bigg\}+\frac{1}{\sqrt{n}}\label{usechannel2},
\end{align}
where \eqref{usescarlett} follows from similar steps leading to \cite[Eq.~(38)]{scarlett2017mismatch} with $K_0$ being a finite constant defined in \cite[Eq.~(58)]{tantomamichel2015}, and~\eqref{usechannel2} follows since for each $(i,j)$, $X^n(i,j)$ is generated according to the same distribution (see \eqref{channelspcodebook} or \eqref{channeliidcodebook}) and we denote $X^n(i,j)$ by the  generic random variable $X^n$.

Recall the definitions of $\rmV_{\rms}(\zeta_\rms,\sigma^2)$ in \eqref{def:sourcedispersion}, $\rmV_{\rmc}^{\rm{sp}}(\zeta_\rmc,P)$ in \eqref{def:channeldispersionspherical}, $\rmV_{\rmc}^{\rm{iid}}(\zeta_\rmc,P)$ in \eqref{def:channeldispersioniid}, $\rho^*(P,\sigma^2,D)$ in \eqref{def:mu} and the choice of $k$ in \eqref{kachsecond}.  We first prove Lemma \ref{ach:essential} when we use spherical codebooks for both source and channel codebooks, i.e., $\dagger=\rm{sp}$ and $\ddagger=\rm{sp}$. For simplicity, let
\begin{align}
\rmV_1&:=4\sigma^4(P+1)^2\Big(\rho^*(P,\sigma^2,D)\rmV_\rms(\zeta_\rms,\sigma^2)+\rmV_\rmc^{\rm{sp}}(\zeta_\rmc,P)\Big)\\
&=\rho^*(P,\sigma^2,D)(P+1)^2(\zeta_\rms-\sigma^4)+\sigma^4(4P+P^2(\zeta_\rmc-1))\label{def:rmV1}.
\end{align}

Using the choice of $M_i$ in \eqref{choosemi}, the definition of $\overline{h}_{\ddagger}(n,i)$ in \eqref{def:ohni} and the result in \eqref{usechannel2}, we obtain that 
\begin{align}
\nn&\sum_{i=1}^{\NT}\Pr \{S^k\in\calT_i \}\overline{h}_{\rm{sp}}(n,i)-\frac{1}{\sqrt{n}}\\*
\nn&\leq\sum_{i=1}^{\NT}\int_{s^k\in\calT_i}
\Pr\bigg\{\frac{(P+1)\|Z^n\|^2-\|X^n+Z^n\|^2}{2(P+1)}\geq n\rmC(P)+\log \Psi_{\rm{sp}}(k,\Upsilon(i))\\*
&\qquad\qquad\qquad\qquad\qquad\qquad\qquad\qquad-\log k-\log (\NT K_0 \sqrt{n}) \bigg\} f_{S^k}(s^k)\rmd s^k\label{spsp:alt}\\
\nn&\leq\sum_{i=1}^{\NT}\int_{s^k\in\calT_i}
\Pr\bigg\{\frac{(P+1)\|Z^n\|^2-\|X^n+Z^n\|^2}{2(P+1)}\geq n\rmC(P)+\log \underline{g}\Big(k,\frac{\|s^k\|^2+\sigma^2}{k}\Big)\\
&\qquad\qquad\qquad\qquad\qquad-\log k-\log (\NT K_0\sqrt{n}) \bigg\} f_{S^k}(s^k)\rmd s^k\label{useskincalti}\\
\nn&\leq\sum_{i=1}^{\NT}\int_{s^k\in\calT_i}
\Pr\bigg\{\frac{(P+1)\|Z^n\|^2-\|X^n+Z^n\|^2}{2(P+1)}\geq n\rmC(P)-(k-1)R_{\rm{sp}}\bigg(\frac{\|s^k\|^2+\sigma^2}{k}\bigg)\\*
&\qquad\qquad\qquad\qquad\qquad+\log \frac{\Gamma(\frac{k+2}{2})}{\sqrt{\pi}\Gamma(\frac{k+1}{2})}-\log k-\log (\NT K_0 \sqrt{n}) \bigg\} f_{S^k}(s^k)\rmd s^k\label{usedefug}\\
\nn&=\sum_{i=1}^{\NT}\int_{s^k\in\calT_i}
\Pr\bigg\{\frac{(P+1)\|Z^n\|^2-\|X^n+Z^n\|^2}{2(P+1)}\geq n\rmC(P)-k\rmR(\sigma^2,D)\\*
& \qquad\qquad\qquad\qquad\qquad-\frac{\|s^k\|^2-(k-1)\sigma^2}{2\sigma^2}+O(\log n)\bigg\} f_{S^k}(s^k)\rmd s^k\label{taylorpsi}\\
&=\Pr\bigg\{S^k\in\bigcup_{i=1}^{\NT}\calT_i,\frac{P\|Z^n\|^2-nP-2\langle X^n,Z^n\rangle}{2(P+1)}\geq n\rmC(P)-k\rmR(\sigma^2,D)-\frac{\|S^k\|^2-k\sigma^2}{2\sigma^2}+O(\log n)\bigg\}\label{usespcbook}\\
&\leq \Pr\bigg\{\frac{P\|Z^n\|^2-nP-2\langle X^n,Z^n\rangle}{2(P+1)}+\frac{\|S^k\|^2-k\sigma^2}{2\sigma^2}\geq n\rmC(P)-k\rmR(\sigma^2,D)+O(\log n)\bigg\}\\
\nn&=\Pr\bigg\{\sigma^2\big(P\|Z^n\|^2-nP-2\langle X^n,Z^n\rangle\big)+(P+1)\big(\|S^k\|^2-k\sigma^2\big)\\*
&\qquad\qquad\qquad\qquad\qquad\geq 2\sigma^2(P+1)\big(n\rmC(P)-k\rmR(\sigma^2,D)+O(\log n)\big)\bigg\}\label{ss:step3-1}\\
&\leq \rmQ\Bigg(\rmQ^{-1}(\varepsilon)+O\bigg(\frac{\log n}{\sqrt{n}}\bigg)\Bigg)+O\bigg(\frac{1}{\sqrt{n}}\bigg)\label{ss:step4},
\end{align}
where \eqref{useskincalti} follows since $\frac{\|s^k\|^2+\sigma^2}{k}\geq \Upsilon(i)$ when $s^k\in\calT_i$, $\Psi_{\rm{sp}}(k,p)\geq \underline{g}(k,p)$ (see\ Claim (i) in Lemma \ref{property}) and $\underline{g}(k,p)$ is decreasing in $p\geq |\sigma^2-2D|$, \eqref{usedefug} follows from the definition of $\underline{g}(\cdot)$ in \eqref{def:ugkz}, \eqref{taylorpsi} follows by i) using the Taylor expansion of $R_{\rm{sp}}(p)$ at $p=\sigma^2$ (see~\eqref{tltobeused}) and the bound on $N$ in \eqref{nt:second}, and ii) noting that $ {\Gamma(\frac{k+2}{2})}/{(k\Gamma(\frac{k+1}{2}))}=O(\frac{1}{\sqrt{k}})$, $k=\Theta(n)$, $\big|\frac{\|s^k\|^2}{k}-\sigma^2\big|^2\leq (\xi+\frac{1}{k})^2\sigma^4=O(\frac{\log k}{k})$ since $s^k\in \Upsilon(i)$ (see \eqref{def:Lambdai}) for some $i\in[1:\NT]$ and $\xi=\sqrt{\frac{\log k}{k}}$ (see~\eqref{def:xi}), and \eqref{ss:step4} follows by applying the Berry-Esseen theorem for functions of independent random variables~\cite[Proposition 1]{iri2015third} and using the choice of $k$ in \eqref{kachsecond}. A detailed proof of \eqref{ss:step4} is given in Appendix~\ref{proofssstep4}.

Next, we prove Lemma \ref{ach:essential} when we use the spherical codebook for the source codebook and the i.i.d. Gaussian codebook for the channel codebook. Compared with the case where $\dagger=\rm{sp}$ and $\ddagger=\rm{sp}$, the proof when $\dagger=\rm{sp}$ and $\ddagger=\rm{iid}$ is exactly the same until \eqref{taylorpsi}. Thus, when $\dagger=\rm{sp}$ and $\ddagger=\rm{iid}$, we have
\begin{align}
\nn&\sum_{i=1}^{\NT}\Pr\{S^k\in\calT_i \}\overline{h}_{\rm{iid}}(n,i)\\
&\leq \Pr\bigg\{S^k\in\calT_i,~\frac{(P+1)\|Z^n\|^2-\|X^n+Z^n\|^2}{2(P+1)}\geq n\rmC(P)-k\rmR(\sigma^2,D)-\frac{\|S^k\|^2-(k+1)\sigma^2}{2\sigma^2}+O(\log n)\bigg\}\label{essential:spiid}\\
&\leq \Pr\bigg\{\frac{(P+1)\|Z^n\|^2-\|X^n+Z^n\|^2}{2(P+1)}\geq n\rmC(P)-k\rmR(\sigma^2,D)-\frac{\|S^k\|^2-(k+1)\sigma^2}{2\sigma^2}+O(\log n)\bigg\}\\
&=\Pr\bigg\{\sum_{i=1}^n \sigma^2(PZ_i^2-X_i^2-2X_iZ_i)+(P+1)(S_i^2-\sigma^2)\geq 2\sigma^2(P+1)\big(n\rmC(P)-k\rmR(\sigma^2,D)+O(\log n)\big)\bigg\}\label{iidsp4use}\\
&\leq \rmQ\Bigg(\rmQ^{-1}(\varepsilon)+O\bigg(\frac{\log n}{\sqrt{n}}\bigg)\Bigg)+O\bigg(\frac{1}{\sqrt{n}}\bigg)\label{sameforiid},
\end{align}
where \eqref{iidsp4use} follows since the channel input $X^n$ is i.i.d. according to $\calN(0,P)$ when the channel codebook is an i.i.d. Gaussian codebook, and \eqref{sameforiid} follows similarly to the proof of \eqref{ss:step4} by applying the Berry-Esseen theorem for functions of independent random variables~\cite[Proposition 1]{iri2015third} and is available in Appendix \ref{proofssstep4}. 

Finally, we prove Lemma \ref{ach:essential} when the source codebook is an i.i.d. Gaussian codebook and the channel codebook is either a spherical or an i.i.d. Gaussian codebook. Using \eqref{choosemi} and \eqref{upp:pijwrong},  similarly to arguments leading to \eqref{taylorpsi}, we obtain that for any $\ddagger\in\rm\{sp,iid\}$,
\begin{align}
\nn&\sum_{i=1}^{\NT}\Pr \{S^k\in\calT_i \}\overline{h}_{\ddagger}(n,i)\\
&\leq \sum_{i=1}^{\NT}\int_{s^k\in\calT_i}
\Pr\bigg\{\frac{(P+1)\|Z^n\|^2-\|X^n+Z^n\|^2}{2(P+1)}\geq n\rmC(P)-k\rmR(\sigma^2,D) \nn\\*
& \qquad\qquad\qquad\qquad -\frac{\|s^k\|^2-(k+1)\sigma^2}{2\sigma^2}+O(\log n) \bigg\} f_{S^k}(s^k)\rmd s^k\label{tayloriid},
\end{align}
When the channel codebook is a spherical codebook, the rest of the proof is exactly the same as the steps in~\eqref{usespcbook} to \eqref{ss:step4}. On the other hand, when the channel codebook is an i.i.d.\ Gaussian codebook, the rest of the proof is exactly the same as \eqref{essential:spiid} to \eqref{sameforiid}.
\end{proof}

\subsection{Ensemble Converse Proof}
\label{sec:proofconversesr}
In the ensemble converse proof, for any $\varepsilon\in(0,1)$ and $(\dagger,\ddagger)\in\rm\{sp,iid\}^2$, we assume for the sake of contradiction that there exists a sequence of $(k,n)$-codes (see\  Definition \ref{def:code}) such that \eqref{eqn:vareps} holds and 
\begin{align}
k=n\rho^*(P,\sigma^2,D)-\sqrt{n\rmV_{\ddagger}(\zeta_\rms,\sigma^2,\zeta_\rmc,P)}\rmQ^{-1}(\varepsilon+\tau)~\mbox{for some }\tau > 0.\label{kconsecond}
\end{align}

Recall the lower bound on the ensemble excess-distortion probability in \eqref{lowfirst2} and the definition of $j(s^k,\hat{\bs}_i)$ in \eqref{def:jskhsi}. For any $(\dagger,\ddagger)\in\rm\{sp,iid\}^2$, we can lower bound each term in the first sum in \eqref{lowfirst2} as follows: for any $i\in[1:\NT]$,
\begin{align}
\nn&\Pr\{S^k\in\calT_i,(\hatI,\hatJ)\neq (i,J)\}\\
&=\int_{s^k\in\calT_i}\bigg(\int_{\hat{\bs}_i} \Pr\big\{(\hatI,\hatJ)\neq (i,j(s^k,\hat{\bs}_i))\big|S^k=s^k,~\hat{\bS}_i=\hat{\bs}_i\big\}
f_{\hat{\bS}_i}(\hat{\bs}_i)\bigg)f_S^k(s^k)\rmd s^k\\
&\geq \int_{s^k\in\calT_i}\ \underline{h}_{\ddagger}(n,i) f_S^k(s^k)\rmd s^k\label{uselowpij}\\
&=\Pr\bigg\{S^k\in\calT_i,\frac{(P+1)\|Z^n\|^2-\|X^n+Z^n\|^2}{2(P+1)}\geq n\rmC(P)-\log M_i+O(\log n)\bigg\}\label{c:step1}.
\end{align}
where \eqref{uselowpij} follows by using the result in \eqref{low:pijwrong}, and \eqref{c:step1} follows similarly to the steps leading to~\cite[Eq. (74)]{scarlett2017mismatch}.

Using \eqref{kconsecond} and \eqref{c:step1}, similarly to steps proving Lemma \ref{ach:essential}, we can prove that for any $(\dagger,\ddagger)\in\rm\{sp,iid\}^2$,
\begin{align}
\Pr\Big\{S^k\notin\bigcup_{i=1}^{\NT}\calT_i\Big\}+\sum_{i=1}^{\NT}\Pr\{S^k\in\calT_i,(\hatI,\hatJ)\neq (i,J)\}&\geq \rmQ\Bigg(\rmQ^{-1}(\varepsilon+\tau)+O\bigg(\frac{\log n}{\sqrt{n}}\bigg)\Bigg)+O\bigg(\frac{1}{\sqrt{n}}\bigg)\label{c:step2}.
\end{align}
The proof of \eqref{c:step2} is given in Appendix \ref{proofcstep2} for completeness.

The following lemma is vital in the converse proof. 
\begin{lemma}
\label{lowerresidual}
For any $(\dagger,\ddagger)\in\rm\{sp,iid\}^2$ and for  all $i\in[1:\NT]$,
\begin{align}
\Pr\big\{d(S^k,\hatS^k(\hatI,\hatJ))\leq D\big|S^k\in\calT_i,~\hatI\neq i\big\}\leq\exp\bigg(-k\Big(\frac{1}{2}\log\frac{\sigma^2}{D}+O\Big(\sqrt{\frac{\log k}{k}}\Big)\Big)\bigg)\label{lowr1},\\
\Pr\big\{d(S^k,\hatS^k(\hatI,\hatJ))\leq D|S^k\in\calT_i,~\hatI=i,~\hatJ\neq J\big\}\leq \exp\bigg(-k\Big(\frac{1}{2}\log\frac{\sigma^2}{D}+O\Big(\sqrt{\frac{\log k}{k}}\Big)\Big)\bigg)\label{lowr2}.
\end{align}
\end{lemma}
%The proof of Lemma \ref{lowerresidual} is given in the end of this subsection. %Lemma \ref{lowerresidual} holds regardless of the choices for both source and channel codebooks.

Combining the results in \eqref{lowfirst2}, \eqref{c:step2} and Lemma \ref{lowerresidual}, we conclude for any $k$ satisfying \eqref{kconsecond}, for any $(\dagger,\ddagger)\in\rm\{sp,iid\}^2$,
\begin{align}
\limsup_{n\to\infty}\rmP_{\rme,k,n}\geq \liminf_{n\to\infty}\rmP_{\rme,k,n}\geq \varepsilon+\tau\label{consecondfinal},
\end{align}
This violates the condition that $\rmP_{\rme,k,n}\leq \varepsilon$. Since $\tau>0$ is arbitrary, we have shown that for any $\varepsilon\in[0,1)$  and any $(\dagger,\ddagger)\in\rm\{sp,iid\}^2$,
\begin{align}
L_{\dagger,\ddagger}^*(\varepsilon)\geq \sqrt{\rmV_{\ddagger}(\zeta_\rms,\sigma^2,\zeta_\rmc,P)}\rmQ^{-1}(\varepsilon).
\end{align}

\begin{proof}[Proof of Lemma \ref{lowerresidual}]
\label{prooflowerresidual}

For simplicity, given $s^k\in\calT_i$ and $\hat{\bs}_i$, we let
\begin{align}
\Pr\{(\hatI,\hatJ)=(\hati,\hatj)|s^k,\hat{\bs}_i\}&:=\Pr\{(\hatI,\hatJ)=(\hati,\hatj)|S^k=s^k,\hat{\bS}_i=\hat{\bs}_i\}.
\end{align}
We first prove \eqref{lowr1}. For any $\dagger\in\rm\{sp,iid\}$, we have
\begin{align}
\nn&\Pr\{S^k\in\calT_i,~\hatI\neq i,~d(S^k,\hatS^k(\hatI,\hatJ))\leq D\}\\
&=\int_{s^k\in\calT_i}\bigg(\sum_{\hati=1}^{\NT}\sum_{\hatj=1}^{M_{\hati}}\int_{\hat{\bs}} \Pr\{(\hatI,\hatJ)=(\hati,\hatj)|s^k,\hat{\bs}\}1\big\{\hati\neq i,~d(s^k,\hatS^k(\hati,\hatj))\leq D\big\}f_{\hat{\bS}}^{\dagger}(\hat{\bs})\rmd \hat{\bs}\bigg)f_S^k(s^k)\rmd s^k\label{eqn:become_prob0}\\
\nn&=\int_{s^k\in\calT_i}\Bigg(\sum_{\hati\in[1:\NT]\setminus\{i\}}\sum_{\hatj=1}^{M_{\hati}} \bigg\{\int_{\hat{\bs}_i}\Pr\{(\hatI,\hatJ)=(\hati,\hatj)|s^k,\hat{\bs}_i\}\label{eqn:onemoreline}\\
&\qquad\qquad\qquad\qquad\qquad\qquad \times \Big(\int_{\hat{\bs}_{\hati}}1\big\{ d(s^k,\hatS^k(\hati,\hatj))\leq D\big\}f_{\hat{\bS}_{\hati}}^{\dagger}(\hat{\bs}_{\hati})\rmd \hat{\bs}_{\hati}\Big)f_{\hat{\bS}_i}^{\dagger}(\hat{\bs}_i)\rmd \hat{\bs}_i\bigg\}\Bigg)f_S^k(s^k)\rmd s^k\\
&=\int_{s^k\in\calT_i}\Bigg(\int_{\hat{\bs}_i}\sum_{\substack{\hati\in[1:\NT]\setminus\{i\}}}\sum_{\hatj=1}^{M_{\hati}}\Pr\{(\hatI,\hatJ)=(\hati,\hatj)|s^k,\hat{\bs}_i\}\Pr_{f_{\hatS^k}^{\dagger}}\{d(s^k,\hatS^k(\hati,\hatj))\leq D\}f_{\hat{\bS}_i}^{\dagger}(\hat{\bs}_i)\rmd \hat{\bs}_i\Bigg)f_S^k(s^k)\rmd s^k \label{eqn:become_prob}\\
&=\int_{s^k\in\calT_i}\Bigg(\int_{\hat{\bs}_i}\sum_{\substack{\hati\in[1:\NT]\setminus\{i\}}}\sum_{\hatj=1}^{M_{\hati}}\Pr\{(\hatI,\hatJ)=(\hati,\hatj)|s^k,\hat{\bs}_i\}f_{\hat{\bS}_i}^{\dagger}(\hat{\bs}_i)\rmd \hat{\bs}_i\Bigg) \Psi_{\dagger}\Big(k,\frac{\|s^k\|^2}{k}\Big)f_S^k(s^k)\rmd s^k\label{cstep1}\\
&\leq \int_{s^k\in\calT_i}\Bigg(\int_{\hat{\bs}_i}\sum_{\substack{\hati\in[1:\NT]\setminus\{i\}}}\sum_{\hatj=1}^{M_{\hati}}\Pr\{(\hatI,\hatJ)=(\hati,\hatj)|s^k,\hat{\bs}_i\}f_{\hat{\bS}_i}^{\dagger}(\hat{\bs}_i)\rmd \hat{\bs}_i\Bigg)\Psi_{\dagger}\big(k,\Upsilon(0)\big)f_S^k(s^k)\rmd s^k\label{cstep2}\\
&\leq \Pr\{S^k\in\calT_i,~\hatI\neq i\}\exp\bigg(-k\Big(\frac{1}{2}\log\frac{\sigma^2}{D}+O\Big(\sqrt{\frac{\log k}{k}}\Big)\Big)\bigg),\label{cstep3}
\end{align}
where \eqref{eqn:onemoreline} follows since i) $(\hatI,\hatJ)$ is independent of all subcodebooks $\hat{\bs}_{\bari}$ for $\bari\in[1:\NT]\setminus\{i\}$ (see Section \ref{sec:channeloutput}), and ii) therefore we can divide the whole codebook $\hat{\bs}$ into subcodebooks and integrate over each codebook separately, \eqref{cstep1} follows by using the definition of $\Psi_{\dagger}(\cdot)$ in \eqref{def:Psikz} and noting that each source codeword is generated independently according to the same distribution~(see \eqref{sourcespcodebook} or \eqref{sourceiidcodebook}), \eqref{cstep2} follows since i) $\Psi_{\dagger}\big(k,p\big)$ is decreasing in $p$ for $p\geq \max\{|\sigma^2-2D|,2D-\sigma^2\}$ and ii) for all $s^k\in\calT_i$, we have $\frac{\|s^k\|^2}{k}\geq \Upsilon(i-1)\geq \Upsilon(0)$, and \eqref{cstep3} follows from the bounds on $\Psi_{\dagger}(\cdot)$ in Lemma \ref{property} and the Taylor expansion of $R_{\dagger}(p)$ at $\sigma^2$ (see \eqref{tltobeused}) similarly to \eqref{taylorpsi} and noting the definitions of $\Upsilon(0)$ in \eqref{def:L0} and $\xi$ in \eqref{def:xi}.

We now prove \eqref{lowr2}. Recall the definition of $j(s^k,\hat{\bs}_i)$ in \eqref{def:jskhsi}.  We have the following lemma  which is proved in Appendix~\ref{proofvitallower}. 
\begin{lemma}
\label{vitallowerresidual}
For any $i\in[1:\NT]$, given any $s^k\in\calT_i$ and any subcodebook $\hat{\bs}_i$, then for any $(\hati,\hatj)$ such that $\hati=i$ and $\hatj\neq j(s^k,\hat{\bs}_i)$, we have that $\Pr\{(\hatI,\hatJ)=(i,\hatj)|s^k,\hat{\bs}_i\}$ depends only on $i$. For brevity, we denote this quantity as $q_i:= \Pr\{(\hatI,\hatJ)=(i,\hatj)|s^k,\hat{\bs}_i\}$.
\end{lemma}

Using Lemma \ref{vitallowerresidual}, we have that for any $i\in[1:\NT]$,
\begin{align}
\nn&\Pr\{S^k\in\calT_i,~\hatI=i,~\hatJ\neq J)\}\\
&=\int_{s^k\in\calT_i}\bigg(\int_{\hat{\bs}_i} \sum_{\substack{\hatj\in[1:M_i]\setminus\{j(s^k,\hat{\bs}_i)\}}}\Pr\{(\hatI,\hatJ)=(i,\hatj)|s^k,\hat{\bs}_i\}f_{\hat{\bS}_i}(\hat{\bs}_i)\rmd \hat{\bs}_i\bigg)f_{S^k}(s^k)\rmd s^k\\
&=\int_{s^k\in\calT_i}\bigg(\int_{\hat{\bs}_i} \sum_{\substack{\hatj\in[1:M_i]\setminus\{j(s^k,\hat{\bs}_i)\}}}q_if_{\hat{\bS}_i}(\hat{\bs}_i)\rmd \hat{\bs}_i\bigg)f_{S^k}(s^k)\rmd s^k\\
&=\int_{s^k\in\calT_i}(M_i-1)q_if_{S^k}(s^k)\rmd s^k\label{cstepmvp}.
\end{align}
Furthermore, we have
\begin{align}
\nn&\Pr\{S^k\in\calT_i,\hatI=i,\hatJ\neq J,~d(S^k,\hatS^k(\hatI,\hatJ))\leq D\}\\
&=\int_{s^k\in\calT_i}\Bigg(\int_{\hat{\bs}_i}\sum_{\hatj=1}^{M_i}\Pr\big\{(\hatI,\hatJ)=(i,\hatj)|s^k,\hat{\bs}_i\big\}1\big\{d(s^k,\hats^k(i,\hatj))\leq D,~\hatj\neq j(s^k,\hat{\bs}_i)\big\}f_{\hat{\bS}_i}^{\dagger}(\hat{\bs}_i)\rmd \hat{\bs}_i \Bigg)f_S^k(s^k)\rmd s^k\\
&=\int_{s^k\in\calT_i}\Bigg(\int_{\hat{\bs}_i}\sum_{\substack{\hatj\in[1:M_i]\setminus\{j(s^k,\hat{\bs}_i)\}}}\Pr\big\{(\hatI,\hatJ)=(i,\hatj)|s^k,\hat{\bs}_i\big\}1\{d(s^k,\hats^k(i,\hatj))\leq D\}f_{\hat{\bS}_i}^{\dagger}(\hat{\bs}_i)\rmd\hat{\bs}_i\Bigg)f_S^k(s^k)\rmd s^k\\
&=\int_{s^k\in\calT_i}\Bigg(\int_{\hat{\bs}_i}\sum_{\substack{\hatj\in[1:M_i]\setminus\{j(s^k,\hat{\bs}_i)\}}}q_i1\{d(s^k,\hats^k(i,\hatj))\leq D\}f_{\hat{\bS}_i}^{\dagger}(\hat{\bs}_i)\rmd\hat{\bs}_i\Bigg)f_S^k(s^k)\rmd s^k\label{usecstep4.0}\\
&\leq \int_{s^k\in\calT_i}\Big(\int_{\hat{\bs}_i}\sum_{\hatj=1}^{M_i}q_i1\{d(s^k,\hats^k(i,\hatj))\leq D\}f_{\hat{\bS}_i}^{\dagger}(\hat{\bs}_i)\rmd \hat{\bs}_i\Big)f_S^k(s^k)\rmd s^k\label{usecstep4}\\
&=\int_{s^k\in\calT_i}M_iq_i\Big(\Pr_{f_{\hatS^k}^{\dagger}}\{d(s^k,\hatS^k(i,\hatj))\leq D\}\Big)f_S^k(s^k)\rmd s^k\\
&\leq\int_{s^k\in\calT_i}M_iq_i\Psi_{\dagger}(k,\Upsilon(0))f_S^k(s^k)\rmd s^k\label{usecstep2}\\
&=\frac{M_i}{M_i-1}\Psi_{\dagger}(k,\Upsilon(0))\Pr\{S^k\in\calT_i,~\hatI=i,~\hatJ\neq J)\}\label{usemvp}\\
&\leq\Pr\{S^k\in\calT_i,~\hatI=i,~\hatJ\neq J\}\exp\bigg(-k\Big(\frac{1}{2}\log\frac{\sigma^2}{D}+O\Big(\sqrt{\frac{\log k}{k}}\Big)\Big)\bigg)\label{cstep6}.
\end{align}
where \eqref{usecstep4.0} follows from Lemma \ref{vitallowerresidual}, \eqref{usecstep2} follows from similar steps leading to \eqref{cstep2}, \eqref{usemvp} follows from \eqref{cstepmvp}, and \eqref{cstep6} follows from i) the fact that $\frac{M_i}{M_i-1}\to 1$ as $k\to\infty$ and ii) a Taylor expansion similar to \eqref{cstep3}.

This completes the proof of Lemma \ref{lowerresidual}.
\end{proof}

\section{Proof of Moderate Deviations Asymptotics (Theorem \ref{jscc:mdc})}
\label{proof:jscc:mdc}
\subsection{Preliminaries}

The following lemma generalizes the moderate deviations theorem (cf. \cite[Theorem 3.7.1]{dembo2009large}) for i.i.d.\ random vectors to smooth functions of independent but not necessarily identically distributed random vectors.

\begin{lemma}
\label{mdc4funcofirv}
Let $\{\bU_i\}_{i=1}^\infty$  be a sequence of independent but not necessarily identically distributed random vectors in $\bbR^d$. Furthermore, let $f:\bbR^d\to\bbR$ be a function with uniformly bounded second derivatives and let $\bJ=(J_1,\ldots,J_d)$ be a row vector of first-order partial derivatives of $f$, i.e.,
\begin{align}
J_r&:=\frac{\partial f(\bu)}{\partial u_r}\bigg|_{\bu=\mathbf{0}},~\forall~r\in[1:d].
\end{align}
Let the components of $\bU_i$ be $(U_{i,1},\ldots,U_{i,d})$ for each $i\in\bbN$. Finally, let
\begin{align}
\rmV_n&:=\mathrm{Cov}\Big(\frac{1}{\sqrt{n}}\sum_{i=1}^n\bU_i\Big),\\
\rmV(t)&:=\lim_{n\to\infty}\frac{1}{n}\sum_{i=1}^n\mathrm{Var}[U_{i,t}]\label{def:rmvtlimit},~\forall~t\in[1:d],
\end{align}Assume that 
\begin{enumerate}
\item There exists some ball around the origin such that  $\bm{\lambda}\mapsto\log \mathbb{E}[\exp(\langle \bm{\lambda},\bU_i\rangle)]$ is finite for all $i\in\bbN$; 
\item There exists some ball around the origin such that    $\lambda\mapsto\log \mathbb{E}[\exp(\lambda U_{i,t})]$ is finite for all $i\in\bbN$ and $t\in [1:d]$;
\item The limit in \eqref{def:rmvtlimit} exists and is positive for all $t\in[1:d]$;
\item The limit $\rmV:=\lim_{n\to\infty}\bJ \rmV_n\bJ^\rmT$ exists and is positive;
\end{enumerate}
we have that for any moderate deviations sequence $\eta_n$ (see \eqref{mdc:constaint}) and any positive number $\alpha$,
\begin{align}
\lim_{n\to\infty}-\frac{1}{n\eta_n^2}\log\Pr\Big\{f\Big(\frac{1}{n}\sum_{i=1}^n\bU_i\Big)\geq f(\mathbf{0})+\alpha\eta_n\Big\}=\frac{\alpha^2}{2\rmV}.
\end{align}
\end{lemma}
The proof of Lemma \ref{mdc4funcofirv} is provided in Appendix \ref{proofmdc4func}.

\subsection{Achievability Proof}
\label{sec:proofmdcach}
For any $(\dagger,\ddagger)\in\rm\{sp,iid\}^2$, let $k$ be chosen such that
\begin{align}
k=n\rho^*(P,\sigma^2,D)-n\eta_n\label{kachmdc},
\end{align}
where $\rho^*(P,\sigma^2,D)$ is defined in \eqref{def:mu}. In particular, $k$ is linear in $n$.

Similarly to the achievability proof of the second-order asymptotics from \eqref{useseveral} to \eqref{ss:step2}, for any $(\dagger,\ddagger)\in\rm\{sp,iid\}^2$, we obtain 
\begin{align}
\rmP_{\rme,k,n}
&\leq \Pr\Big\{S^k\notin\bigcup_{i=1}^{\NT}\calT_i\Big\}
+\sum_{i=1}^{\NT} \int_{s^k\in\calT_i}\exp(-M_i\Psi_{\dagger}(k,\|s^k\|^2/k))f_{S^k}(s^k) \rmd s^k\nn\\*
&\qquad+\sum_{i=1}^{\NT}\int_{s^k\in\calT_i} \overline{h}_{\ddagger}(n,i) f_{S^k}(s^k) \rmd s^k\\
&\leq \Pr\Big\{S^k\notin\bigcup_{i=1}^{\NT}\calT_i\Big\}
+\exp(-k)+\sum_{i=1}^{\NT}\int_{s^k\in\calT_i} \overline{h}_{\ddagger}(n,i) f_{S^k}(s^k) \rmd s^k\label{usemimdc},
\end{align}
where \eqref{usemimdc} follows by using the choice of $M_i$ in \eqref{choosemi} and arguments used to prove \eqref{ss:step3}.

Using the definition of $\xi$ in \eqref{def:ximdc} and the moderate deviations theorem~\cite[Theorem 3.7.1]{dembo2009large}, under condition \eqref{assump1} in Theorem~\ref{jscc:mdc}, we obtain that
\begin{align}
\lim_{n\to\infty}-\frac{1}{n\eta_n^{3/2}}\log \Pr\Big\{S^k\notin\bigcup_{i=1}^{\NT}\calT_i\Big\} 
&=\frac{\rho^*(P,\sigma^2,D)}{2(\zeta_\rms-\sigma^4)}\label{mdcproofs1}.
\end{align}

The following lemma is essential in the achievability proof.
\begin{lemma}
\label{mdc:achessential}
For any $(\dagger,\ddagger)\in\rm\{sp,iid\}^2$, with the choice of $k$ in \eqref{kachmdc} and under the conditions in Theorem \ref{jscc:mdc}, we have
\begin{align}
\lim_{n\to\infty}-\frac{1}{n\eta_n^2}\log \bigg(\sum_{i=1}^{\NT}\int_{s^k\in\calT_i} \overline{h}_{\ddagger}(n,i) f_{S^k}(s^k) \rmd s^k\bigg)
&=\frac{1}{2\rmV_{\ddagger}(\zeta_\rms,\sigma^2,\zeta_\rmc,P)}.
\end{align}
\end{lemma}
The proof of Lemma \ref{mdc:achessential} is deferred to Appendix \ref{proof:mdcachess}.

Combining the results in \eqref{usemimdc}, \eqref{mdcproofs1} and Lemma \ref{mdc:achessential} and noting that the third term in \eqref{usemimdc} dominates, we conclude that for any for any $(\dagger,\ddagger)\in\rm\{sp,iid\}^2$, there exists a sequence of $(k,n)$-codes satisfying \eqref{kachmdc} and
\begin{align}
\liminf_{n\to\infty}-\frac{1}{n\eta_n^2}\log \rmP_{\rme,k,n}&\geq\frac{1}{2\rmV_{\ddagger}(\zeta_\rms,\sigma^2,\zeta_\rmc,P)}.
\end{align}

\subsection{Ensemble Converse Proof}
The general procedure of the ensemble converse proof of moderate deviations is to show that for any $(\dagger,\ddagger)\in\rm\{sp,iid\}^2$ and any sequence of $(k,n)$-codes such that 
\begin{align}
k=\lceil n\rho^*(P,\sigma^2,D)-n\eta_n\rceil \label{kconmdc},
\end{align}
we have
\begin{align}
\liminf_{n\to\infty}-\frac{1}{n\eta_n^2}\log \rmP_{\rme,k,n}
&\leq \limsup_{n\to\infty}-\frac{1}{n\eta_n^2}\log \rmP_{\rme,k,n}\leq\frac{1}{2\rmV_{\ddagger}(\zeta_\rms,\sigma^2,\zeta_\rmc,P)}\label{mdcviolate}. 
\end{align}
Indeed, the analysis of the term in \eqref{c:step1} uses Lemma \ref{mdc4funcofirv} instead of the Berry-Essen theorem for functions of random vectors. This can be done similarly to the proof of  Lemma \ref{mdc:achessential} in Section \ref{sec:proofmdcach}.

\appendix
\subsection{Proof of Lemma \ref{channelerror}}
\label{proofchannelerror}
Recall the definition of the  mismatched  information density in \eqref{def:mismatchid} (cf. \cite[Eqn. (28)-(29)]{scarlett2017mismatch}). We first prove \eqref{upp:pijwrong}. 
Given the channel output $Y^n$, the output of the channel decoder in~\eqref{decoder}  is a pair $(\hatI,\hatJ)$ that maximizes a generalized mismatched information density, i.e.,
\begin{align}
(\hatI,\hatJ)&=\argmax_{(\tili,\tilj) \in\calD} \imath(X^n(\tili,\tilj);Y^n)-\log M_{\tili}.
\end{align}

Recall the definition of $j(s^k,\hat{\bs}_i)$ in \eqref{def:jskhsi}. The following steps mimic the proof of the RCU bound in~\cite{polyanskiy2010finite}. For any $i\in[1:\NT]$, conditioned on the events $\{\hat{\bS}_i=\hat{\bs}_i\}$ and $\{S^k=s^k\}$  for some $s^k\in\calT_i$, we can upper bound the probability of decoding $(i,j(s^k,\hat{\bs}_i))$ incorrectly  as follows:
\begin{align}
\nn&\Pr\big\{(\hatI,\hatJ)\neq (i,j(s^k,\hat{\bs}_i))|S^k=s^k,~\hat{\bS}_i=\hat{\bs}_i,\big\}\\
&=\Pr \left\{\exists~(\tili,\tilj)\in\calD\setminus(i,j(s^k,\hat{\bs}_i)): \imath(X^n(\tili,\tilj);Y^n)-\log M_{\tili}\leq \imath(X^n(i,j(s^k,\hat{\bs}_i));Y^n)-\log M_i \right\}\label{useskhsi}\\
&\leq\sum_{(\tili,\tilj)\in\calD\setminus \{ (i,j(s^k,\hat{\bs}_i)) \} }\mathbb{E}\bigg[\Pr\bigg\{\imath(X^n(\tili,\tilj);Y^n)\leq \imath(X^n(i,j(s^k,\hat{\bs}_i));Y^n)-\log \frac{M_i}{M_{\tili}}\bigg|X^n(i,j(s^k,\hat{\bs}_i)),Y^n\bigg\}\bigg]\\
\nn&=\Bigg(\sum_{\tili\in[1:|\NT]\setminus\{i\}}M_{\tili}\mathbb{E}\bigg[\Pr\bigg\{\imath(\barX^n;Y^n)\leq \imath(X^n;Y^n)-\log \frac{M_i}{M_{\tili}}\bigg|X^n,Y^n\bigg\}\bigg]\Bigg)\\*
&\qquad+(M_i-1)\mathbb{E}\bigg[\Pr\bigg\{\imath(\barX^n;Y^n)\leq \imath(X^n;Y^n)-\log \frac{M_i}{M_{\tili}}\bigg|X^n,Y^n\bigg\}\bigg]\label{cerror1}\\
&\leq \sum_{\tili=1}^{\NT}M_{\tili}\mathbb{E}\bigg[\Pr\bigg\{\imath(\barX^n;Y^n)\leq \imath(X^n;Y^n)-\log \frac{M_i}{M_{\tili}}\bigg|X^n,Y^n\bigg\}\bigg],\label{cerror11}
\end{align}
where \eqref{useskhsi} follows since when $\hatS^k=s^k$ and $\hat{\bS}_i=\hat{\bs}_i$ for $s^k\in\calT_i$, the probability that $(\hatI,\hatJ)\neq (i,j(s^k,\hat{\bs}_i))$ equals the probability that the channel decoder fails to decode the message pair $(i,j(s^k,\hat{\bs}_i))$ correctly, and \eqref{cerror1} follows since each channel codeword is generated independently according to the same distribution (either \eqref{channelspcodebook} or \eqref{channeliidcodebook}) and we use $\barX^n$ and $X^n$ to denote the generic random variables $X^n(\tili,\tilj)$ and $X^n(i,j(s^k,\hat{\bs}_i))$ respectively. Hence, $(\barX^n,X^n,Y^n)$ is distributed according to~\eqref{def:joint}.

We also always have $\Pr\big\{(\hatI,\hatJ)\neq (i,j)|S^k=s^k,~\hat{\bS}_i=\hat{\bs}_i\big\}\leq 1$. % for any $(i,j)\in\calD$, $s^k$ and $\hat{\bs}_i$. 
Hence, using \eqref{cerror11}, for any $i\in[1:\NT]$, given $s^k\in\calT_i$ and the subcodebook $\hat{\bs}_i$, we have 
\begin{align}
\nn&\Pr\big\{(\hatI,\hatJ)\neq (i,j(s^k,\hat{\bs}_i))|S^k=s^k,~\hat{\bS}_i=\hat{\bs}_i\big\}\\
&\leq \min\Bigg\{1,\mathbb{E}\Bigg[\sum_{\tili=1}^{\NT}M_{\tili}\Pr\bigg\{\imath(\barX^n;Y^n)\leq \imath(X^n;Y^n)-\log \frac{M_i}{M_{\tili}}\bigg|X^n,Y^n\bigg\}\Bigg]\Bigg\}\label{barx6ninde}\\
&=\overline{h}_{\ddagger}(n,i).\label{useohni}
\end{align}

Next we   prove \eqref{low:pijwrong}. %Recall our channel decoder in \eqref{decoder}. 
For any $i\in[1:\NT]$, given $s^k\in\calT_i$ and the subcodebook $\hat{\bs}_i$, following similar steps to prove \eqref{cerror11}, we obtain
\begin{align}
\nn&\Pr\Big\{(\hatI,\hatJ)\neq (i,j(s^k,\hat{\bs}_i))|S^k=s^k,~\hat{\bS}_i=\hat{\bs}_i\Big\}\\
&=1-\Pr\Big\{\|X^n(i,j(s^k,\hat{\bs}_i))-Y^n\|^2+2\log M_i\leq \|X^n(\tili,\tilj)-Y^n\|^2+2\log M_{\tili},~\forall~(\tili,\tilj)\in\calD\setminus\{(i,j(s^k,\hat{\bs}_i))\}\Big\} \label{eqn:conditioning}\\
&\geq 1-\Pr\Big\{\|X^n(i,j(s^k,\hat{\bs}_i))-Y^n\|^2\leq \|X^n(i,\tilj)-Y^n\|^2,~\forall~\tilj\in[1:M_i]\setminus\{j(s^k,\hat{\bs}_i)\}\Big\}\\
&=1-\Big(1-\Pr\big\{\|\barX^n-Y^n\|^2\leq \|X^n-Y^n\|\big\}\Big)^{M_i-1}\label{barx6nagain}\\
&=1-\bigg(1-\mathbb{E}\Big[\Pr\Big\{\|\barX^n-Y^n\|^2\leq \|X^n-Y^n\|\big|X^n,Y^n\Big\}\Big]\bigg)^{M_i-1}\label{useuhni0}\\
&=\underline{h}_{\ddagger}(n,i).\label{useuhni}
\end{align}

\subsection{Proof of \eqref{ss:step4} and \eqref{sameforiid}}
\label{proofssstep4}
Similarly to \cite{scarlett2017mismatch}, we will apply the Berry-Esseen theorem for functions of independent random variables. This constitutes a special case of \cite[Proposition 1]{iri2015third} when the Markov chain is of order zero (see also \cite[Proposition 1]{molavianjazi2015second} for the i.i.d.\ case). 

We   first prove \eqref{ss:step4}. Let $\tilX^n\sim\calN(\mathbf{0},\mathbf{I}_n )$. Then a spherically-distributed channel codeword $X^n$ (see \eqref{channelspcodebook}) can be written as 
\begin{align}
X^n=\sqrt{nP}\frac{\tilX^n}{\|\tilX^n\|}.
\end{align}
For $i\in[1:n]$, let $A_{4,i}:=0$ and
\begin{align}
A_{1,i}:=Z_i^2-1,\quad
A_{2,i} :=\sqrt{P}\tilX_iZ_i,\quad
A_{3,i} :=\tilX_i^2-1,
\end{align}
and for $i\in[n+1:n+k]$, let $A_{1,i}=A_{2,i}=A_{3,i}=0$ and let
\begin{align}
A_{4,i}=S_{i-n}^2-\sigma^2.
\end{align}
Furthermore, let 
\begin{align}
\gamma_{\rm{sp}}(a_1,a_2,a_3,a_4)&:=\sigma^2\bigg(Pa_1-\frac{2a_2}{\sqrt{1+a_3}}\bigg)+(P+1)a_4.
\end{align}
Then, it can be verified that
\begin{align}
\nn&(n+k)\gamma_{\rm{sp}}\Big(\frac{1}{n+k}\sum_{i=1}^{n+k}A_{1,i},\frac{1}{n+k}\sum_{i=1}^{n+k}A_{2,i},\frac{1}{n+k}\sum_{i=1}^{n+k}A_{3,i},\frac{1}{n+k}\sum_{i=1}^{n+k}A_{4,i}\Big)\\*
&=\sigma^2\big(P\|Z^n\|^2-nP-2\langle X^n,Z^n\rangle\big)+(P+1)\big(\|S^k\|^2-k\sigma^2\big)\label{ftouse}.
\end{align}
Note that the vector of partial derivatives of the function $\gamma_{\rm{sp}}(\cdot)$ evaluated at $\mathbf{0}$ is $\mathbf{J}=[\sigma^2P,-2\sigma^2,0,P+1]$. The covariance matrix is
\begin{align}
\mathbf{V}_{n+k}&=\mathrm{Cov}\bigg(\frac{1}{\sqrt{n+k}}\sum_{i=1}^{n+k}A_{1,i},\frac{1}{\sqrt{n+k}}\sum_{i=1}^{n+k}A_{2,i},\frac{1}{\sqrt{n+k}}\sum_{i=1}^{n+k}A_{3,i},\frac{1}{\sqrt{n+k}}\sum_{i=1}^{n+k}A_{4,i}\bigg)\\
&=\mathrm{diag}\bigg(\Big[\frac{n(\zeta_\rmc-1)}{n+k},\frac{nP}{n+k},\frac{2n}{n+k},\frac{k(\zeta_\rms-\sigma^4)}{n+k}\Big]\bigg).
\end{align} 
Thus,
\begin{align}
\mathbf{J}\mathbf{V}_{n+k}\mathbf{J}^{\rmT}
&=\frac{n\sigma^4(4P+P^2(\zeta_\rmc-1))+k(P+1)^2(\zeta_\rms-\sigma^4)}{n+k}.
\end{align}
Recalling the choice of $k$ in \eqref{kachsecond} and using the definition of $\rmV_1$ in \eqref{def:rmV1}. we conclude that
\begin{align}
\lim_{n\to\infty}\mathbf{J}\mathbf{V}_{n+k}\mathbf{J}^{\rmT}
&=\frac{\rmV_1}{1+\rho^*(P,\sigma^2,D)}\label{factn1}.
\end{align}
Therefore, we obtain
\begin{align}
\nn&\Pr\Big\{\sigma^2\big(P\|Z^n\|^2-nP-2\langle X^n,Z^n\rangle\big)+(P+1)\big(\|S^k\|^2-k\sigma^2\big)\geq 2\sigma^2(P+1)\big(n\rmC(P)-k\rmR(\sigma^2,D)+O(\log n)\big)\Big\}\\*
\nn&=\Pr\Big\{\sigma^2\big(P\|Z^n\|^2-nP-2\langle X^n,Z^n\rangle\big)+(P+1)\big(\|S^k\|^2-k\sigma^2\big)\\*
&\qquad\qquad\qquad\geq 2\sigma^2(P+1)\rmR(\sigma^2,D)\big(\sqrt{n\rmV_{\rm{sp}}(\zeta_\rms,\sigma^2,\zeta_\rmc,P)}\rmQ^{-1}(\varepsilon)+O(\log n)\big)\Big\}\label{usekacha2.0}\\
\nn&=\Pr\bigg\{(n+k)\gamma_{\rm{sp}}\Big(\frac{1}{n+k}\sum_{i=1}^{n+k}A_{1,i},\frac{1}{n+k}\sum_{i=1}^{n+k}A_{2,i},\frac{1}{n+k}\sum_{i=1}^{n+k}A_{3,i},\frac{1}{n+k}\sum_{i=1}^{n+k}A_{4,i}\Big)\\
&\qquad\qquad\qquad\geq 2\sigma^2(P+1)\rmR(\sigma^2,D)\big(\sqrt{n\rmV_{\rm{sp}}(\zeta_\rms,\sigma^2,\zeta_\rmc,P)}\rmQ^{-1}(\varepsilon)+O(\log n)\big)\bigg\}\label{usekacha2}\\
&\leq \rmQ\Bigg(\frac{2\sigma^2(P+1)\rmR(\sigma^2,D)\big(\sqrt{n\rmV_{\rm{sp}}(\zeta_\rms,\sigma^2,\zeta_\rmc,P)}\rmQ^{-1}(\varepsilon)+O(\log n)\big)}{\sqrt{\frac{n+k}{1+\rho^*(P,\sigma^2,D)}\rmV_1}}\Bigg)+O\bigg(\frac{1}{\sqrt{n+k}}\bigg)\label{usenfact1}\\
&=\rmQ\Bigg(\frac{\sqrt{n\rmV_{\rm{sp}}(\zeta_\rms,\sigma^2,\zeta_\rmc,P)}\rmQ^{-1}(\varepsilon)+O(\log n)}{\sqrt{n\rmV_{\rm{sp}}(\zeta_\rms,\sigma^2,\zeta_\rmc,P)}\sqrt{\frac{1+\rho^*(P,\sigma^2,D)-\sqrt{{\rmV_{\rm{sp}}(\zeta_\rms,\sigma^2,\zeta_\rmc,P)}/{n}}\rmQ^{-1}(\varepsilon)}{1+\rho^*(P,\sigma^2,D)}}}\Bigg)+O\bigg(\frac{1}{\sqrt{n}}\bigg)\label{usekagain}\\
&=\rmQ\Bigg(\rmQ^{-1}(\varepsilon)+O\bigg(\frac{\log n}{\sqrt{n}}\bigg)\Bigg)+O\bigg(\frac{1}{\sqrt{n}}\bigg)\label{hhafinal},
\end{align}
where \eqref{usekacha2.0} follows by using the choice of $k$ in \eqref{kachsecond} and noting that $k=\Theta(n)$, \eqref{usekacha2} follows by using \eqref{ftouse}, \eqref{usenfact1} follows from invoking the Berry-Esseen theorem for functions of independent random variables~\cite[Proposition 1]{iri2015third} and  \eqref{factn1}, and \eqref{usekagain} follows from the choice of $k$ in \eqref{kachsecond} and the  definitions of $\rmV_1$ in \eqref{def:rmV1} and $\rmV_{\rm{sp}}(\zeta_\rms,\sigma^2,\zeta_\rmc,P)$ in \eqref{def:rmvdagger}. 

To prove \eqref{sameforiid}, the function we use is $\gamma_{\rm{iid}}(b_1,b_2,b_3,b_4):=\sigma^2(Pb_1-b_2-2b_3)+(P+1)b_4$ and the random variables are defined as follows: for $i\in[1:n]$, let $B_{1,i}:=Z_i^2$, $B_{2,i}:=X_i^2$, $B_{3,i}:=X_iZ_i$ and $B_{4,i}=0$ while for $i\in[n+1:n+k]$, let $B_{1,i}=B_{2,i}=B_{3,i}:=0$ and $B_{4,i}:=S_i^2-\sigma^2$. It can thus be verified that
\begin{align}
\nn&(n+k)\gamma_{\rm{iid}}\Big(\frac{1}{n+k}\sum_{i=1}^{n+k}B_{1,i},\frac{1}{n+k}\sum_{i=1}^{n+k}B_{2,i},\frac{1}{n+k}\sum_{i=1}^{n+k}B_{3,i},\frac{1}{n+k}\sum_{i=1}^{n+k}B_{4,i}\Big)\\
&=\sum_{i=1}^n \sigma^2(PZ_i^2-X_i^2-2X_iZ_i)+(P+1)(S_i^2-\sigma^2)\label{factn3}.
\end{align}
Furthermore, the vector of partial derivatives of $\gamma_{\rm{iid}}(\cdot)$ evaluated at $\mathbf{0}$ is $[\sigma^2P,-\sigma^2,-2\sigma^2,P+1]$ and the covariance matrix is
\begin{align}
\bV_{n+k}'
&=\mathrm{Cov}\bigg(\frac{1}{\sqrt{n+k}}\sum_{i=1}^{n+k}B_{1,i},\frac{1}{\sqrt{n+k}}\sum_{i=1}^{n+k}B_{2,i},\frac{1}{\sqrt{n+k}}\sum_{i=1}^{n+k}B_{3,i},\frac{1}{\sqrt{n+k}}\sum_{i=1}^{n+k}B_{4,i}\bigg)\\
&=\mathrm{diag}\bigg(\Big[\frac{n(\zeta_\rmc-1)}{n+k},\frac{2nP^2}{n+k},\frac{Pn}{n+k},\frac{k(\zeta_\rms-\sigma^4)}{n+k}\Big]\bigg).
\end{align}
Therefore, it can be verified that
\begin{align}
\lim_{n\to\infty}\bJ\rmV_{n+k}'\bJ^{\rmT}
&=\frac{\rho^*(P,\sigma^2,D)(P+1)^2(\zeta_\rms-\sigma^4)+\sigma^4\Big(P^2(\zeta_\rmc+1)+4P\Big)}{1+\rho^*(P,\sigma^2,D)}\label{factn2}.
\end{align}
The rest of the proof of \eqref{sameforiid} is omitted since it is similar to steps in \eqref{usekacha2.0} to \eqref{hhafinal} by applying the Berry-Esseen theorem for functions of random vectors and applying the facts in \eqref{factn3} and \eqref{factn2}.

\subsection{Proof of \eqref{c:step2}}
\label{proofcstep2}
For simplicity, we prove \eqref{c:step2} only when both the source and the channel codebooks are spherical codebooks since other cases can be proved similarly.

Using \eqref{c:step1}, we have that for any $i\in[1:\NT]$,
\begin{align}
\nn&\Pr\{S^k\in\calT_i,(\hatI,\hatJ)\neq (i,J)\}\\
&\geq\Pr\bigg\{S^k\in\calT_i,~\frac{(P+1)\|Z^n\|^2-\|X^n+Z^n\|^2}{2(P+1)}\geq n\rmC(P)-\log M_i+O(\log n)\bigg\}\\
&=\Pr\bigg\{S^k\in\calT_i,~\frac{(P+1)\|Z^n\|^2-\|X^n+Z^n\|^2}{2(P+1)}\geq n\rmC(P)+\log\Psi_{\rm{sp}}(k,\Upsilon(i))-\log k+O(\log n)\bigg\}\label{usemisecondagain}\\
&\geq \Pr\bigg\{S^k\in\calT_i,~\frac{(P+1)\|Z^n\|^2-\|X^n+Z^n\|^2}{2(P+1)}\geq n\rmC(P)+\log\overline{g}(k,\Upsilon(i))-\log k+O(\log n)\bigg\}\label{uselammasp}\\
&\geq \Pr\bigg\{S^k\in\calT_i,~\frac{(P+1)\|Z^n\|^2-\|X^n+Z^n\|^2}{2(P+1)}\geq n\rmC(P)+\log\overline{g}\Big(k,\frac{\|S^k\|^2}{k}\Big)-\log k+O(\log n)\bigg\}\label{uselammasp2}\\
&=\Pr\bigg\{S^k\in\calT_i,~\frac{(P+1)\|Z^n\|^2-\|X^n+Z^n\|^2}{2(P+1)}\geq n\rmC(P)-k\rmR(\sigma^2,D)-\frac{\|S^k\|^2-k\sigma^2}{2\sigma^2}+O(\log n)\bigg\}\label{tayloragain2}\\
&=\Pr\bigg\{S^k\in\calT_i,~\frac{(P+1)\|Z^n\|^2-\|X^n+Z^n\|^2}{2(P+1)\rmR(\sigma^2,D)}+\frac{\|S^k\|^2-k\sigma^2}{2\sigma^2\rmR(\sigma^2,D)}\geq \sqrt{n\rmV_{\ddagger}(\zeta_\rms,\sigma^2,\zeta_\rmc,P)}\rmQ^{-1}(\varepsilon+\tau)+O(\log n)\bigg\}\label{usekconverse},
\end{align}
where \eqref{usemisecondagain} follows from the choice of $M_i$ in \eqref{choosemi}, \eqref{uselammasp} follows since i) $\Psi_{\rm{sp}}(k,p)\leq \overline{g}(k,p)$ for $p\geq 2D-\sigma^2$ and ii) for $k$ sufficiently large, $\Upsilon(i)\geq 2D-\sigma^2$ for all $i\in[1:\NT]$, \eqref{uselammasp2} follows from the facts that i) $\overline{g}(k,p)$ is decreasing in $p$ (implied by conclusion (ii)-b) in Lemma \ref{property}) and ii) for $s^k\in\calT_i$ (see \eqref{def:type}), we have $\frac{\|s^k\|^2}{k}\leq \Upsilon(i)$, \eqref{tayloragain2} follows similarly to arguments leading to \eqref{taylorpsi}, and \eqref{usekconverse} follows from the choice of $k$ in \eqref{kconsecond}.

Using \eqref{usekconverse}, we obtain that
\begin{align}
\nn&\Pr\Big\{S^k\notin\bigcup_{i=1}^{\NT}\calT_i\Big\}+\sum_{i=1}^{\NT}\Pr\{S^k\in\calT_i,(\hatI,\hatJ)\neq (i,J)\}\\
&\geq \Pr\bigg\{\frac{(P+1)\|Z^n\|^2-\|X^n+Z^n\|^2}{2(P+1)\rmR(\sigma^2,D)}+\frac{\|S^k\|^2-k\sigma^2}{2\sigma^2\rmR(\sigma^2,D)}\geq \sqrt{n\rmV_{\ddagger}(\zeta_\rms,\sigma^2,\zeta_\rmc,P)}\rmQ^{-1}(\varepsilon+\tau)+O(\log n)\bigg\}\\
\nn&=\Pr\bigg\{\sigma^2\big(P\|Z^n\|^2-nP-2\langle X^n,Z^n\rangle\big)+(P+1)\Big(\|S^k\|^2-k\sigma^2\Big)\\
&\qquad\qquad\qquad\geq 2\sigma^2(P+1)\rmR(\sigma^2,D)\Big(\sqrt{n\rmV_{\ddagger}(\zeta_\rms,\sigma^2,\zeta_\rmc,P)}\rmQ^{-1}(\varepsilon+\tau)\Big)+O(\log n)\bigg\}\label{spchannelcb}\\
&\geq \rmQ\Bigg(\rmQ^{-1}(\varepsilon+\tau)+O\bigg(\frac{\log n}{\sqrt{n}}\bigg)\Bigg)+O\bigg(\frac{1}{\sqrt{n}}\bigg)\label{simtoachsecond},
\end{align}
where \eqref{spchannelcb} follows since each codeword is generated independently and uniformly over a sphere with radius $\sqrt{nP}$ when we use a spherical codebook as the channel codebook, and \eqref{simtoachsecond} follows by using the Berry-Esseen theorem for functions of independent random variables~\cite[Proposition 1]{iri2015third} similarly to Appendix \ref{proofssstep4} and   details are thus omitted.

\subsection{Proof of Lemma \ref{vitallowerresidual}}
\label{proofvitallower}
Recall the notation and results in Section \ref{sec:channeloutput} and the definition of $j(s^k,\hat{\bs}_i)$ in \eqref{def:jskhsi}. For brevity, given two pairs $(\hati,\hatj)\in\calD$ and $(\tili,\tilj)\in\calD$, define the event
\begin{align}
\calA_{\hati,\hatj,\tili,\tilj}&:=\bigg\{\|X^n(\hati,\hatj)-Y^n\|^2\leq \|X^n(\tili,\tilj)-Y^n\|^2+2\log\frac{M_{\tili}}{M_{\hati}}\bigg\}\label{def:cerrorevent}.
\end{align}
Then for any $i\in[1:\NT]$, given $s^k\in\calT_i$ and the subcodebook $\hat{\bs}_i$, for any $(\hati,\hatj)\in\calD\setminus\{(i,j(s^k,\hat{\bs}_i))\}$, we have
\begin{align}
\nn&\Pr\{(\hatI,\hatJ)=(\hati,\hatj)|s^k,\hat{\bs}_i\}\\
&=\Pr\Big\{\bigcap_{(\tili,\tilj)\in\calD\setminus\{(\hati,\hatj)\}}\calA_{\hati,\hatj,\tili,\tilj}\Big|s^k,\hat{\bs}_i\Big\}\\
&=\mathbb{E}\bigg[\prod_{(\tili,\tilj)\in\calD\setminus\{(\hati,\hatj)\}}1\{\calA_{\hati,\hatj,\tili,\tilj}\}\Big|s^k,\hat{\bs}_i\bigg]\\
&=\mathbb{E}\Bigg[\mathbb{E}\bigg[\bigg(\prod_{(\tili,\tilj)\in\calD\setminus\{(\hati,\hatj),(i,j(s^k,\hat{\bs}_i))\}}1\{\calA_{\hati,\hatj,\tili,\tilj}\}\bigg)\times 1\{\calA_{\hati,\hatj,i,j(s^k,\hat{\bs}_i)}\}\bigg|X^n(\hati,\hatj),X^n(i,j(s^k,\hat{\bs}_i)),Z^n\bigg]\Bigg]\\ 
&=\mathbb{E}\Bigg[\bigg(\prod_{(\tili,\tilj)\in\calD\setminus\{(\hati,\hatj),(i,j(s^k,\hat{\bs}_i))\}}\Pr\Big\{\calA_{\hati,\hatj,\tili,\tilj}\Big|X^n(\hati,\hatj),X^n(i,j(s^k,\hat{\bs}_i)),Z^n\Big\}\bigg)\times 1\{\calA_{\hati,\hatj,i,j(s^k,\hat{\bs}_i)}\}\Bigg]\label{defphatij},
\end{align}
where \eqref{defphatij} holds since each channel codeword is generated independently. Using \eqref{defphatij}, given $s^k\in\calT_i$ and $\hat{\bs}_i$ for any $i\in[1:\NT]$, for any $(\hati,\hatj)\in\calD$ such that $\hati=i$ and $\hatj\neq j(s^k,\hat{\bs}_i)$, we have
\begin{align}
\nn&\Pr\{(\hatI,\hatJ)=(\hati,\hatj)|s^k,\hat{\bs}_i\}\\
&=\mathbb{E}\Bigg[\bigg(\prod_{(\tili,\tilj)\in\calD\setminus\{(i,\hatj),(i,j(s^k,\hat{\bs}_i))\}}\Pr\Big\{\calA_{\hati,\hatj,\tili,\tilj}\Big|X^n(i,\hatj),X^n(i,j(s^k,\hat{\bs}_i)),Z^n\Big\}\bigg)\times 1\{\calA_{i,\hatj,i,j(s^k,\hat{\bs}_i)}\}\Bigg]\\
\nn&=\mathbb{E}\Bigg[\bigg(\prod_{\tili\in[1:\NT]\setminus\{i\}}\prod_{\tilj\in[1:M_i]}\Pr\Big\{\calA_{\hati,\hatj,\tili,\tilj}\Big|X^n(i,\hatj),X^n(i,j(s^k,\hat{\bs}_i)),Z^n\Big\}\bigg)\\*
&\qquad\times\bigg(\prod_{\tilj\in[1:M_i]\setminus\{\hatj,j(s^k,\hat{\bs}_i)\}}\Pr\Big\{\calA_{\hati,\hatj,i,\tilj}\Big|X^n(i,\hatj),X^n(i,j(s^k,\hat{\bs}_i)),Z^n\Big\}\bigg)\times 1\{\calA_{i,\hatj,i,j(s^k,\hat{\bs}_i)}\}\Bigg]\\
\nn&=\mathbb{E}\Bigg[\bigg\{\prod_{\substack{\tili\in[1:\NT]\setminus\{i\}}}\bigg(\Pr\Big\{\|\hatX^n-Y^n\|^2\leq \|\barX^n-Y^n\|^2+2\log\frac{M_{\tili}}{M_i}\Big|\hatX^n,X^n,Z^n\Big\}\bigg)^{M_{\tili}}\bigg\}\\*
&\quad\qquad\times \bigg(\Pr\Big\{\|\hatX^n-Y^n\|^2\leq \|\barX^n-Y^n\|^2\Big|\hatX^n,X^n,Z^n\Big\}\bigg)^{M_i-2} \times 1\Big\{\|\hatX^n-Y^n\|^2\leq \|X^n-Y^n\|^2\Big\}\bigg]\Bigg]\label{cstep4},
\end{align} 
where \eqref{cstep4} follows from the definition of $\calA_{\hati,\hatj,\tili,\tilj}$ in \eqref{def:cerrorevent} and similar arguments used to obtain \eqref{cerror1}.  Note that given any $i\in[1:\NT]$, any $s^k\in\calT_i$ and any $\hat{\bs}_i$, the right hand side of \eqref{cstep4} depends only on $i$ (i.e., it does not depend either $(s^k,\hat{\bs}_i)$ or $(\hati,\hatj)$) as long as $\hati=i$ and $\hatj\neq j(s^k,\hat{\bs}_i)$. The proof of Lemma \ref{vitallowerresidual} is now complete.

\subsection{Proof of Lemma \ref{mdc4funcofirv}}
\label{proofmdc4func}
The proof of Lemma \ref{mdc4funcofirv} is inspired by \cite[Proposition 1]{iri2015third} and makes use of the following lemma.
\begin{lemma}
\label{mdc4irvs}
Let $\{U_i\}_{i=1}^\infty$ be a sequence of independent but not necessarily identically distributed zero mean random variables satisfying the following two conditions:
\begin{enumerate}
\item  There exists some ball $\calF$ around the origin such that for all $i\in \bbN$, $\lambda \in\calF\mapsto\Lambda_{U_i}(\lambda)$ is finite.
\item The limit  $\rmV :=\lim_{n\to\infty}\frac{1}{n}\sum_{i=1}^n\mathrm{Var}(U_i)$ exists and is positive.
\end{enumerate}
For any moderate deviations sequence $\eta_n$ (see \eqref{mdc:constaint}) and any positive number $\alpha$, we have
\begin{align}
\lim_{n\to\infty} -\frac{1}{n\eta_n^2}\log \Pr\Big\{\frac{1}{n}\sum_{i=1}^n U_i>\eta_n\alpha\Big\}
&=\frac{\alpha^2}{2\rmV}.
\end{align}
\end{lemma}
Lemma \ref{mdc4irvs} is a straightforward generalization of \cite[Theorem 3.7.1]{dembo2009large} to independent but not necessarily identically distributed random variables using G\"artner-Ellis Theorem (cf. \cite[Theorem 2.3.6]{dembo2009large}) and also appeared in \cite{hayashi2015asymmetric}. The proof of Lemma \ref{mdc4irvs} is thus omitted.

Define the typical set
\begin{align}
\calB&:=\Big\{\bu^n:\Big\|\frac{1}{n}\sum_{i=1}^n\bu_i\Big\|_{\infty}\leq \eta_n^{3/4}\Big\}\label{def:calb}.
\end{align}

For any $\bu^n\in\calB$, Taylor expanding $f(\cdot)$  at $\mathbf{0}$ and noting that the second-order derivatives of the function $f$ are uniformly bounded, we obtain that
\begin{align}
f\Big(\frac{1}{n}\sum_{i=1}^n\bu_i\Big)
&=f(\mathbf{0})+\frac{1}{n}\sum_{i=1}^n\langle \bJ,\bu_i\rangle+o(\eta_n)\label{taylor4mdc}.
\end{align}
Thus, using \eqref{taylor4mdc}, we obtain that 
\begin{align}
\Pr\Big\{f\Big(\frac{1}{n}\sum_{i=1}^n\bU_i\Big)\geq f(\mathbf{0})+\alpha\eta_n\Big\}
&\leq \Pr\Big\{f\Big(\frac{1}{n}\sum_{i=1}^n\bU_i\Big)\geq f(\mathbf{0})+\alpha\eta_n,~\bU^n\in\calB\Big\}+\Pr\{\bU^n\notin\calB\}\\
&=\Pr\Big\{\frac{1}{n}\sum_{i=1}^n\langle\bJ,\bU_i\rangle\geq \alpha\eta_n+o(\eta_n),~\bU^n\in\calB\Big\}+\Pr\{\bU^n\notin\calB\}\\
&\leq \Pr\Big\{\frac{1}{n}\sum_{i=1}^n\langle\bJ,\bU_i\rangle\geq \alpha\eta_n+o(\eta_n)\Big\}+\Pr\{\bU^n\notin\calB\}\label{mdcstep0}.
\end{align}
Similarly, 
\begin{align}
\Pr\Big\{f\Big(\frac{1}{n}\sum_{i=1}^n\bU_i\Big)\geq f(\mathbf{0})+\alpha\eta_n\Big\}
&\geq \Pr\Big\{\frac{1}{n}\sum_{i=1}^n\langle\bJ,\bU_i\rangle\geq \alpha\eta_n+o(\eta_n)\Big\}-\Pr\{\bU^n\notin\calB\}\label{usea&b}.
\end{align}

Note that 
\begin{align}
\bJ\rmV_n\bJ^\rmT
&=\frac{1}{n}\sum_{i=1}^n\mathrm{Var}[\langle\bJ,\bU_i\rangle]\label{mdcstep2}.
\end{align}
Therefore, using \eqref{mdcstep2} and Lemma \ref{mdc4irvs}, we conclude that under conditions (i) and (iv)  in Lemma \ref{mdc4funcofirv},
\begin{align}
\lim_{n\to\infty}-\frac{1}{n\eta_n^2}\log \Pr\Big\{\frac{1}{n}\sum_{i=1}^n\langle\bJ,\bU_i\rangle\geq \alpha\eta_n+o(\eta_n)\Big\}&=\frac{\alpha^2}{2\rmV}\label{mdcstep3}.
\end{align}
In the rest of the proof, we   upper bound $\Pr\{\bU^n\notin\calB\}$. Using the definition of $\calB$ in \eqref{def:calb}, we obtain that
\begin{align}
\Pr\{\bU^n\notin\calB\}
%&=\Pr\Big\{\Big\|\frac{1}{n}\sum_{i=1}^n\bU_i\Big\|_{\infty}>\eta_n^{3/4}\Big\}\\
&\leq \sum_{t=1}^d \Pr\Big\{\Big|\frac{1}{n}\sum_{i=1}^nU_{i,t}\Big|>\eta_n^{3/4}\Big\}\label{upp14use}.
\end{align}

Recall the definition of $\rmV(t)$ in \eqref{def:rmvtlimit}. Following similar steps to prove \cite[Theorem~3.7.1]{dembo2009large} and using \eqref{upp14use}, one can show that under conditions (ii) and (iii) in Lemma \ref{mdc4funcofirv}, % \red{should the block below be a lower bound on the exponent?} \blue{Yes.}
\begin{align}
\liminf_{n\to\infty}-\frac{1}{n\eta_n^{3/2}}\log \Pr\{\bU^n\notin\calB\}
&\geq \min_{t\in[1:d]}\liminf_{n\to\infty}-\frac{1}{n\eta_n^{3/2}}\log\Pr\Big\{\Big|\frac{1}{n}\sum_{i=1}^nU_{i,t}\Big|>\eta_n^{3/4}\Big\}\label{mdcstep4_} \\
&=\frac{1}{\max_{t\in[1:d]}2\rmV(t)}\label{mdcstep4}.
\end{align} 
Hence, the term $\Pr\{\bU^n\notin\calB\}$ is asymptotically negligible.
The proof of Lemma \ref{mdc4funcofirv} is completed by combining~\eqref{mdcstep0},~\eqref{usea&b},~\eqref{mdcstep3},  and~\eqref{mdcstep4}.

\subsection{Proof of Lemma \ref{mdc:achessential}}
\label{proof:mdcachess}
The proof of Lemma \ref{mdc:achessential} is similar to that of Lemma \ref{ach:essential} except that we use  Lemma \ref{mdc4funcofirv} instead of   Berry-Esseen Theorems. 

Using the definition of $\overline{h}_{\ddagger}(n,i)$ in \eqref{def:ohni}, we obtain that for any $\ddagger\in\rm\{sp,iid\}$, we obtain that
\begin{align}
\overline{h}_{\ddagger}(n,i)
&\leq \Pr\bigg\{\frac{(P+1)\|Z^n\|^2-\|X^n+Z^n\|^2}{2(P+1)}\geq n\rmC(P)-\log M_i-\log \big(\NT K_0\exp(n\eta_n^{3/2})\big)\bigg\}+\exp(-n\eta_n^{3/2})\label{toexplainoeta} .
\end{align}
This can be done similar to the steps from~\eqref{usescarlett} to \eqref{usechannel2} except that we  replace $\frac{1}{\sqrt{n}}$ by $\exp(-n\eta_n^{3/2})$.

We first consider the  case when we use spherical codebooks for both source and channel codebooks.  
Recalling the definition of $\gamma_{\rm{sp}}(\cdot)$ and the definitions of random variables $\{(A_{1,i} ,A_{2,i},A_{3,i},A_{4,i})\}_{i\in[1:n+k]}$ in Appendix  \ref{proofssstep4}, we have that
\begin{align}
\nn&\sum_{i=1}^{\NT}\int_{s^k\in\calT_i} \overline{h}_{\ddagger}(n,i) f_{S^k}(s^k) \rmd s^k-\exp(-n\eta_n^{3/2})\\
\nn&\leq \Pr\Big\{\sigma^2\big(P\|Z^n\|^2-nP-2\langle X^n,Z^n\rangle\big)+(P+1)\big(\|S^k\|^2-k\sigma^2\big)\\*
&\qquad\qquad\qquad\qquad\qquad\qquad\qquad\qquad\geq 2\sigma^2(P+1)\big(n\rmC(P)-k\rmR(\sigma^2,D)+o(n\eta_n)\Big\}\label{oetanforreasons}\\
\nn&=\Pr\bigg\{(n+k)\gamma_{\rm{sp}}\Big(\frac{1}{n+k}\sum_{i=1}^{n+k}A_{1,i},\frac{1}{n+k}\sum_{i=1}^{n+k}A_{2,i},\frac{1}{n+k}\sum_{i=1}^{n+k}A_{3,i},\frac{1}{n+k}\sum_{i=1}^{n+k}A_{4,i}\Big)\\*
&\qquad\qquad\qquad\qquad\qquad\qquad\qquad\qquad \geq 2\sigma^2(P+1)\rmR(\sigma^2,D)n\big(\eta_n+o\big(\eta_n\big)\big)\bigg\}\label{useftouse}\\
\nn&=\Pr\bigg\{\gamma_{\rm{sp}}\Big(\frac{1}{n+k}\sum_{i=1}^{n+k}A_{1,i},\frac{1}{n+k}\sum_{i=1}^{n+k}A_{2,i},\frac{1}{n+k}\sum_{i=1}^{n+k}A_{3,i},\frac{1}{n+k}\sum_{i=1}^{n+k}A_{4,i}\Big)\\*
&\qquad\qquad\qquad\qquad\qquad\qquad\qquad\qquad \geq \frac{2\sigma^2(P+1)\rmR(\sigma^2,D)}{1+\rho^*(P,\sigma^2,D)}\big(\eta_n+o(\eta_n)\big)\bigg\}\label{useconsmdc}.
\end{align}
where \eqref{oetanforreasons} follows since $\eta_n^{3/2}=o(\eta_n)$, $\log N=o(\log n)$ (see~\eqref{def:Ntypes} and \eqref{kachmdc}) and $\frac{\log n}{n}=o(\eta_n)$ (see \eqref{mdc:constaint}), \eqref{useftouse} follows from the definitions of $A_{1,i}$, $A_{2,i}$, $A_{3,i}$ and $A_{4,i}$ in Appendix  \ref{proofssstep4}, the result in \eqref{ftouse} and the choice of $k$ in \eqref{kachmdc}, and \eqref{useconsmdc} follows from the fact that
\begin{align}
\frac{n(\eta_n+o(\eta_n))}{n+k}
&=\frac{\eta_n+o(\eta_n)}{1+\rho^*(P,\sigma^2,D)}. \label{eqn:simplify_eta}
\end{align}
From the definitions of $A_{1,i}$, $A_{2,i}$, $A_{3,i}$ and $A_{4,i}$ in Appendix  \ref{proofssstep4}, we conclude that the conditions of Lemma \ref{mdc4funcofirv} are all satisfied. Thus, using \eqref{useconsmdc} and Lemma \ref{mdc4funcofirv}, we see that the proof of Lemma \ref{mdc:achessential} is completed for the case when $\dagger=\rm{sp}$ and $\dagger=\rm{sp}$.

Next, we consider the proof of Lemma \ref{mdc:achessential} when we use a spherical codebook for the source codebook and an i.i.d.\  Gaussian codebook for the channel codebook. The proof of Lemma \ref{mdc:achessential} for this case differs from the proof of Lemma \ref{ach:essential} only in the analysis of the probability term in \eqref{iidsp4use} and it can be done by using Lemma \ref{mdc4funcofirv}. The proofs for the other two cases are similar and thus omitted.

\bibliographystyle{IEEEtran}
\bibliography{IEEEfull_lin}

\end{document}